\documentclass[journal]{IEEEtran}

\usepackage{graphicx}
\usepackage{amsmath,amssymb}

\usepackage[utf8]{inputenc}
\usepackage{tikz}
\usetikzlibrary{shapes.geometric, arrows.meta, positioning, calc, shadows}
\usepackage{nomencl}
\usepackage{etoolbox}
\usepackage{algorithm} 
\usepackage{algpseudocode} 
\usepackage[T1]{fontenc}
\usepackage{cite}
\usepackage{array, multirow}
\usepackage[dvipsnames]{xcolor}
\makenomenclature
\usepackage{url}
\usepackage{threeparttable}
\usepackage{hyperref}
\usepackage{comment}
\usepackage{amsthm}
\newtheorem{proposition}{Proposition}
\newtheorem{remark}{Remark}
\graphicspath{{PNG/}}
\ifCLASSOPTIONcompsoc
  \usepackage[labelfont=sf,textfont=sf]{subfig}
\else
  \usepackage{subfig}
\fi

\usepackage[
    top=0.5in,
    bottom=0.5in,
    left=0.5in,
    right=0.5in
]{geometry}

\usepackage[none]{hyphenat} 
\usepackage{microtype}      
\title{Nested-Loop Trajectory-Informed Variational Quantum Solver for Interior-Point OPF}
\author{Farshad Amani, \textit{Graduate Student Member, IEEE}, Amin Kargarian, \textit{Senior Member, IEEE}
\date{Feb 2024}
\thanks{This work was supported by the National Science Foundation under Grants ECCS-1944752 and ECCS-2312086.

The authors are with the Department of Electrical and Computer Engineering, Louisiana State University, Baton Rouge, LA 70803 USA (e-mail: famani1@lsu.edu, kargarian@lsu.edu).
}
}
\begin{document}
\maketitle
\begin{abstract}
Optimal power flow (OPF) solved by an interior-point method (IPM) requires repeatedly solving  Newton linear systems. When variational quantum linear solvers (VQLS) are used, each IPM iteration involves an additional nested inner variational optimization loop, which can significantly slow the overall quantum-assisted IPM convergence. To address this challenge, this paper proposes a dual-level trainable quantum IPM framework for OPF that leverages early solver-generated trajectories rather than relying on single-point prediction. The key observation is that early IPM iterates provide informative primal-dual, slack, and barrier-variable evolution about the path to optimality, while early VQLS parameter updates provide useful information about the later variational search. At the quantum-solver level, a trainable parameter model uses a short prefix of the VQLS parameter trajectory to project the remaining variational search toward a lower-cost region. At the OPF-solver level, a second trainable model uses early primal-dual IPM iterates to project a later central path state, which is restored to an admissible point before IPM refinement continues. Simulation studies show that the proposed approach reduces the number of variational updates by up to $95\%$ while maintaining OPF objective values close to the classical IPM reference. A 2-bus demonstration on real quantum hardware is also included to validate the implementation of the proposed workflow.
\end{abstract} 
\begin{IEEEkeywords}
Optimal power flow, quantum interior-point method, coherent variational quantum linear solver, trajectory-informed projection, trainable initialization.
\end{IEEEkeywords}
\section*{Nomenclature}
\addcontentsline{toc}{section}{Nomenclature}

\begin{IEEEdescription}[\IEEEusemathlabelsep\IEEEsetlabelwidth{$\mathbf{S}_{m,k}^{(q)}(K_{\mathrm{var}})$}]

\item[\emph{Sets, Indices, and Counters:}]
\item[$\mathcal{G}$] Generator set.
\item[$i$] Generator index.
\item[$M$] Number of OPF scenarios.
\item[$r$] \text{(C)}VQLS update index.
\item[$l$] LCU unitary-term index.
\item[$q$] Solver index, $q\in\{\mathrm{VQLS},\mathrm{CVQLS}\}$.
\item[$K_{\mathrm{var}}$] \text{(C)}VQLS prefix length.
\item[$\mathcal{K}_{\mathrm{var}}$] Candidate \text{(C)}VQLS prefix set.
\item[$K_{\mathrm{ipm}}$] IPM prefix length.
\item[$K_{\mathrm{var}}^\star$] Selected \text{(C)}VQLS prefix length.
\item[$K_{\mathrm{var}}^{\max}$] Maximum \text{(C)}VQLS prefix length.
\item[$t^\star$] Selected IPM target iteration.

\item[\emph{OPF and IPM Quantities:}]
\item[$\mathbf{x}$] OPF decision vector.
\item[$f(\mathbf{x})$] OPF objective function.
\item[$\mathbf{g}(\mathbf{x})$] Equality constraints.
\item[$\mathbf{h}(\mathbf{x})$] Inequality constraints.
\item[$P_{G_i}$] Active generation of unit $i$.
\item[$a_i,b_i,c_i$] Cost coefficients of unit $i$.
\item[$\mathbf{z}$] Primal-dual IPM vector.
\item[$\mathbf{F}(\mathbf{z})$] KKT residual function.
\item[$\Delta \mathbf{z}^{(k)}$] Newton step at iteration $k$.
\item[$\boldsymbol{\lambda}$] Equality-constraint dual vector.
\item[$\boldsymbol{\nu}$] Inequality-constraint dual vector.
\item[$\mathbf{s}$] Slack-variable vector.
\item[$\mu$] Barrier parameter.

\item[\emph{Linear-System Quantities:}]
\item[$\mathbf{A}$] Linear-system matrix.
\item[$\mathbf{b}$] Right-hand-side vector.
\item[$\mathbf{d}^{(k)}$] Newton direction at iteration $k$.
\item[$\mathbf{A}_{m,k}$] System matrix for scenario $m$, iteration $k$.
\item[$\mathbf{d}_{m,k}$] Newton direction for scenario $m$, iteration $k$.
\item[$\mathbf{b}_{m,k}$] Right-hand side for scenario $m$, iteration $k$.

\item[\emph{\text{(C)}VQLS Quantities:}]
\item[$U_b$] State-preparation unitary.
\item[$|d(\boldsymbol{\theta})\rangle$] Ansatz trial state for the search direction.
\item[$p_q$] Number of parameters for solver $q$.
\item[$L$] Number of LCU terms.
\item[$|\Psi(\boldsymbol{\theta})\rangle$] Normalized state proportional to $\mathbf{A}|d(\boldsymbol{\theta})\rangle$.
\item[$\epsilon_{\mathrm{q}}$] \text{(C)}VQLS tolerance.
\item[$d_a$] Ansatz depth.
\item[$\tau_q$] Average per-update computational cost for solver $q$.

\item[\emph{Learning-Based \text{(C)}VQLS Quantities:}]
\item[$\Theta_{m,k}^{(q)}$] Parameter trajectory for solver $q$.
\item[$\boldsymbol{\theta}_{m,k}^{(q),\mathrm{tar}}$] Best-cost target parameter vector.
\item[$\mathbf{S}_{m,k}^{(q)}(K_{\mathrm{var}})$] Variational-parameter prefix sequence.
\item[$\hat{\boldsymbol{\theta}}_{m,k}^{(q),\mathrm{tar}}$] LSTM-predicted parameter target.
\item[$\mathcal{D}_{\mathrm{var,tr}}$] \text{(C)}VQLS training set.
\item[$\mathcal{D}_{\mathrm{var,val}}$] \text{(C)}VQLS validation set.
\item[$N_{\mathrm{tr}}$] Number of training samples.
\item[$N_{\mathrm{val}}$] Number of validation samples.
\item[$\mathcal{L}_{\mathrm{var}}^{(q)}(\varphi_q)$] \text{(C)}VQLS LSTM loss.
\item[$E_{\theta}(K_{\mathrm{var}})$] Prefix prediction error.
\item[$\mathcal{J}_{\mathrm{pre}}(K_{\mathrm{var}})$] Prefix-selection score.

\item[\emph{Learning-Based IPM Central Path Quantities:}]
\item[$\boldsymbol{\Omega}_{\mathrm{ipm}}$] Weighted norm matrix for IPM prediction error.
\item[$r_{\mathrm{sv}}$] Rank parameter for dominant right-singular subspace.
\item[$\mathbf{z}_{m,k}$] IPM state for scenario $m$, iteration $k$.
\item[$\mathcal{Z}_m$] IPM trajectory for scenario $m$.
\item[$\mathbf{P}_{m}(K_{\mathrm{ipm}})$] IPM prefix sequence.
\item[$\mathbf{z}_{m}^{\mathrm{tar}}$] Target IPM point.
\item[$g_{\psi}^{\mathrm{ipm}}(\cdot)$] IPM-level LSTM predictor.
\item[$\psi$] IPM-level LSTM weights.
\item[$\hat{\mathbf{z}}_{m}^{\mathrm{tar}}$] Predicted IPM target point.
\item[$\boldsymbol{\eta}_{k}$] LSTM hidden state.
\item[$\boldsymbol{\chi}_{k}$] LSTM cell state.
\item[$\mathbf{W}_{\mathrm{ipm}}$] IPM readout weight matrix.
\item[$\boldsymbol{\beta}_{\mathrm{ipm}}$] IPM readout bias vector.
\item[$\mathcal{L}_{\mathrm{ipm}}(\psi)$] IPM-level LSTM loss.
\item[$\mathcal{R}(\cdot)$] IPM-state restoration operator.

\end{IEEEdescription}
\section{Introduction}

\IEEEPARstart{O}{ptimal} power flow (OPF) remains a central problem in power system operation because it determines operating points that satisfy network equations and engineering limits while optimizing an objective such as generation cost, loss, or voltage performance \cite{morstyn2022annealing}. Its computational importance has grown as modern grids incorporate larger shares of renewable generation, distributed energy resources, storage, and other sources of operational variability \cite{ullah2022quantum, mastroianni2023assessing}. Under these conditions, OPF is no longer only a modeling problem; it is also a computational problem, since useful solutions must be obtained within limited time and under changing system conditions \cite{pan2020deepopf, amani2025quantum}. This has motivated continued interest in algorithmic and computational methods to reduce the cost of solving large-scale OPF problems without weakening the physical structure of the formulation \cite{chen2025review, amani2025optimal}.

Among available methods, interior-point methods (IPMs) remain a standard choice for large-scale OPF because they handle nonlinear equality and inequality constraints and perform reliably in practice. However, each IPM iteration requires solving a large, sparse linear system associated with the Newton or KKT step \cite{bienstock2020mathematical}. For realistic OPF instances, these linear systems are often symmetric, indefinite, and ill-conditioned, making the linear-algebra stage a major fraction of total runtime. The repeated solution of these systems therefore remains a key computational bottleneck in OPF solvers \cite{amani2025optimal}, and addressing it requires methods that can exploit the structure of the Newton systems without abandoning the outer IPM framework \cite{amani2024quantum}.

This bottleneck has motivated growing interest in quantum methods for power system optimization \cite{rajabi2026distributed, morstyn2024opportunities}. The HHL algorithm established the theoretical possibility of quantum speedups for linear systems, but its circuit depth and fault-tolerance requirements limit near-term applicability \cite{amani2025optimal}. Variational approaches are more compatible with noisy intermediate-scale quantum hardware, replacing deep circuits with a hybrid quantum-classical loop based on parameterized circuits and classical optimization \cite{liu2024quantum, hasanzadeh2026distributed}. In particular, the variational quantum linear solver (VQLS) and its coherent variant (CVQLS) seek to prepare a quantum state proportional to the solution of a linear system by iteratively minimizing a cost function over circuit parameters \cite{kaseb2024quantum}. When \text{(C)}VQLS\footnote{\text{(C)}VQLS denotes either VQLS or CVQLS throughout, unless one solver is specifically identified.} is embedded within OPF IPM, it is invoked repeatedly across IPM iterations, and each invocation incurs an inner variational loop \cite{pareek2025limitations}. The efficiency of this loop depends strongly on initialization and the regularity of the optimization landscape encountered at each Newton step \cite{hasan2021hybrid, amani2025quantum}.

A critical but underexploited observation is that both the outer IPM loop and the inner \text{(C)}VQLS loop generate \emph{trajectories}---ordered sequences of iterates whose evolution encodes predictive structure about the direction and rate of convergence. At the OPF level, the primal-dual variables trace a central path whose early iterates already contain directional information about where the solution lies \cite{amani2025learning}. At the quantum-solver level, the \text{(C)}VQLS parameter updates trace a descent path whose early steps reflect the loss landscape geometry for that particular Newton system \cite{puig2025variational, morales2024quantum}. Crucially, because IPM-generated Newton systems along the central path share closely related numerical structure \cite{amani2025learning}, their associated \text{(C)}VQLS optimization trajectories exhibit corresponding regularity across successive IPM iterations. These properties motivate learning directly from trajectory prefixes---that is, from the actual behavior of the solver on the current problem instance---rather than from static problem features alone, which is the strategy adopted by single-shot predictors \cite{amani2026learning, ardali2026deep}.

This distinction is fundamental to the proposed approach. A single-shot predictor maps problem data directly to a predicted solution or parameter set, without observing how the solver behaves on the instance at hand. By contrast, a trajectory-informed predictor conditions its prediction on a prefix of solver iterates, capturing instance-specific directional information that static features cannot provide. The early iterations serve as a compressed, solver-generated summary of the current problem's convergence behavior, and this information substantially sharpens the prediction. At the OPF level, the predicted primal-dual point is restored to an admissible state and refined by IPM, preserving the feasibility-enforcing role of the optimizer. At the quantum-solver level, the predicted \text{(C)}VQLS parameters serve as an informed initialization, after which variational optimization continues if further refinement is needed \cite{hasan2021hybrid, amani2025quantum}. Machine learning has been explored in OPF for solution approximation and warm-starting \cite{amani2026learning, kaseb2024quantum}, but the nested, trajectory-informed structure proposed here---accelerating both loops simultaneously using solver-generated prefixes---has not been previously developed \cite{amani2025learning}.

In this paper, we develop a nested-loop trajectory-informed Q-IPM framework for OPF. At the inner quantum-solver level, an LSTM model uses an early prefix of the \text{(C)}VQLS parameter trajectory to project the variational search toward a lower-cost region, reducing the parameter-update effort for each IPM-generated Newton system. At the outer OPF-solver level, a second LSTM model uses early primal-dual IPM iterates to project a later central path state, which is then restored to an admissible primal-dual point before IPM refinement. LSTM is preferred over attention-based sequence models because its recurrent structure naturally captures the sequential dependence among solver iterates, where earlier steps inform subsequent ones, making it well-suited to short solver trajectories and limited data. The result is a framework that accelerates both the inner variational loop and the outer IPM loop while retaining the correction and feasibility-enforcing properties of the underlying solvers.

\section{Q-IPM for OPF Solution}
\label{sec:qipm_prelim}

The baseline Q-IPM structure is shown in Fig.~\ref{fig:qipm_nested_loops}. The outer IPM loop forms Newton systems, and the inner \text{(C)}VQLS loop solves them to recover the Newton directions. Thus, the computational effort depends on both outer IPM iterations and inner variational updates.
\begin{figure}[!t]
    \centering
    \includegraphics[width=0.9\columnwidth]{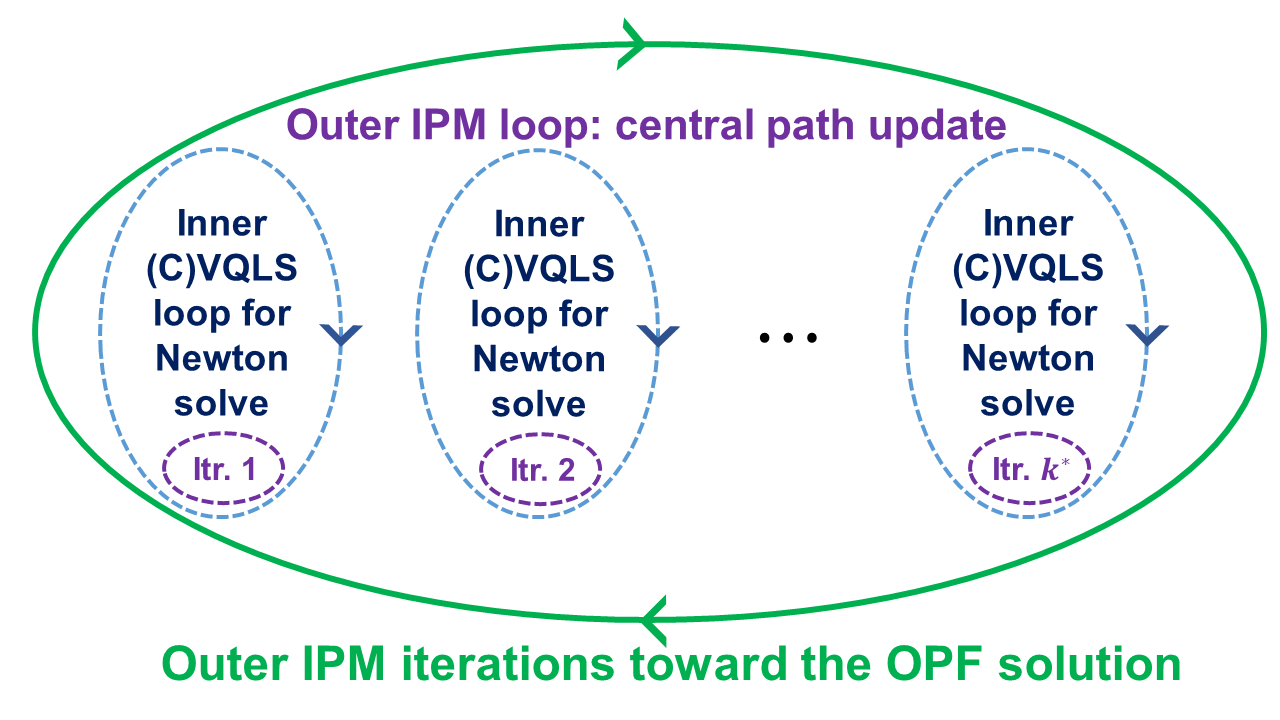}
    \caption{Nested Q-IPM structure for OPF}
    \label{fig:qipm_nested_loops}
\end{figure}
\subsection{OPF Formulation}

The OPF problem is written in compact form as \cite{zimmerman2016matpower}:
{\small
\begin{equation}
\min_{\mathbf{x}} \; f(\mathbf{x})
\label{eq:opf_obj_general}
\end{equation}}
subject to
{\small
\begin{align}
\mathbf{g}(\mathbf{x}) &= \mathbf{0},
\label{eq:opf_eq_general}\\
\mathbf{h}(\mathbf{x}) &\leq \mathbf{0},
\label{eq:opf_ineq_general}
\end{align}}
The objective is generation cost minimization:{\small
\begin{equation}
f(\mathbf{x})=
\sum_{i\in\mathcal{G}}
\left(a_i P_{G_i}^2+b_i P_{G_i}+c_i\right),
\label{eq:opf_cost}
\end{equation}}

In the AC OPF setting, $\mathbf{x}$ includes active and reactive generation, voltage magnitudes, and voltage angles. In the DC OPF setting, the formulation is simplified by neglecting reactive-power and voltage-magnitude variables, so that $\mathbf{x}$ mainly includes active-power generation and voltage angles.


\subsection{Outer-Loop IPM Newton Linear Systems}

The OPF problem is solved using an IPM. After introducing slack variables and barrier terms, the first-order optimality conditions are written in compact KKT form as
{\small
\begin{equation}
\mathbf{F}(\mathbf{z}) = \mathbf{0}.
\label{eq:kkt_compact}
\end{equation}}
At IPM iteration $k$, a Newton step is computed from
{\small
\begin{equation}
\mathbf{J}(\mathbf{z}^{(k)})\,\Delta \mathbf{z}^{(k)} = -\mathbf{F}(\mathbf{z}^{(k)}).
\label{eq:newton_kkt}
\end{equation}}
So, each IPM iteration requires the solution of a sparse linear system of the form
{\small
\begin{equation}
\mathbf{A}^{(k)} \mathbf{d}^{(k)} = \mathbf{b}^{(k)},
\label{eq:linear_system_general}
\end{equation}}
where $\mathbf{A}^{(k)}$ is the Newton matrix, $\mathbf{d}^{(k)}$ is the Newton search direction, and $\mathbf{b}^{(k)}$ is the corresponding mismatch vector.

Since \eqref{eq:linear_system_general} is solved repeatedly along the IPM central path, the total computational effort depends strongly on the cost of the linear-solve stage. In practical OPF problems, $\mathbf{A}^{(k)}$ is typically sparse and may become ill-conditioned, making repeated Newton solves a major computational bottleneck.


\subsection{Inner-Loop \text{(C)}VQLS Linear-System Solver}

The generic linear system solved by the inner quantum routine is written as \cite{morales2024quantum}
{\small
\begin{equation}
\mathbf{A}\mathbf{d}=\mathbf{b},
\label{eq:vqls_linear_system}
\end{equation}}
where $\mathbf{d}$ denotes the search direction to be recovered, and the normalized right-hand side is encoded as a quantum state $|b\rangle = U_b|0\rangle$. The goal of (C)VQLS is to prepare a quantum state proportional to the normalized Newton direction solving \eqref{eq:vqls_linear_system}.

A parameterized quantum circuit, or ansatz, is used to generate a trial state
{\small
\begin{equation}
|d(\boldsymbol{\theta})\rangle = V(\boldsymbol{\theta})|0\rangle,
\label{eq:vqls_trial_state}
\end{equation}}
where $V(\boldsymbol{\theta})$ denotes the ansatz and $\boldsymbol{\theta}\in\mathbb{R}^{p}$ is the vector of variational parameters. The choice of ansatz determines both the expressibility of the trial state and the trainability of the variational optimization.

To enable circuit-based evaluation, the system matrix is written as a linear combination of unitaries,
{\small
\begin{equation}
\mathbf{A}=\sum_{l=0}^{L-1} c_l \mathbf{A}_l,
\label{eq:lcu_form}
\end{equation}}
where $c_l \in \mathbb{R}$ and each $\mathbf{A}_l$ is unitary. In VQLS, the parameters are updated so that the ansatz state satisfies $\mathbf{A}|d(\boldsymbol{\theta})\rangle \propto |b\rangle$. A common global cost function is
{\small
\begin{equation}
C_{\mathrm{VQLS}}(\boldsymbol{\theta})
=
1-
\frac{
\left|
\langle b|\mathbf{A}|d(\boldsymbol{\theta})\rangle
\right|^2
}{
\langle d(\boldsymbol{\theta})|\mathbf{A}^\dagger \mathbf{A}|d(\boldsymbol{\theta})\rangle
},
\label{eq:vqls_cost}
\end{equation}}
which decreases as $\mathbf{A}|d(\boldsymbol{\theta})\rangle$ becomes aligned with $|b\rangle$. In practice, the quantities in \eqref{eq:vqls_cost} are estimated from quantum measurements, typically through expectation values associated with the unitary terms in \eqref{eq:lcu_form}, while a classical optimizer updates $\boldsymbol{\theta}$ iteratively to reduce the cost.

CVQLS uses the same variational state in \eqref{eq:vqls_trial_state}, but evaluates the linear-system mismatch differently. Define
{\small
\begin{equation}
|\Psi(\boldsymbol{\theta})\rangle
=
\frac{\mathbf{A}|d(\boldsymbol{\theta})\rangle}
{\sqrt{\langle d(\boldsymbol{\theta})|\mathbf{A}^\dagger \mathbf{A}|d(\boldsymbol{\theta})\rangle}},
\label{eq:cvqls_psi}
\end{equation}}
and consider the cost
{\small
\begin{equation}
C_{\mathrm{CVQLS}}(\boldsymbol{\theta})
=
1-
\left|
\langle b|\Psi(\boldsymbol{\theta})\rangle
\right|^2.
\label{eq:cvqls_cost}
\end{equation}}
The key distinction is that CVQLS attempts to implement the action of $\mathbf{A}$ coherently, typically by using ancillary qubits and controlled operations associated with the unitary components of \eqref{eq:lcu_form}. This allows the overlap in \eqref{eq:cvqls_cost} to be estimated more directly from a coherently prepared state. By contrast, standard VQLS usually evaluates the cost through a decomposition into separate expectation values linked to the terms of $\mathbf{A}$ rather than through direct coherent preparation of $\mathbf{A}|d(\boldsymbol{\theta})\rangle$ \cite{CoherentVar}.

Fig.~\ref{fig:qipm_nested_loops} highlights the nested computation in Q-IPM: the outer IPM loop advances the OPF primal-dual state, while each outer iteration invokes an inner \text{(C)}VQLS solve to recover the Newton direction. 
\section{Trainable \text{(C)}VQLS Parameter Initialization}
\label{sec:learning_vqls}

When \text{(C)}VQLS is embedded inside Q-IPM, the overall convergence can be slowed by two nested iterative processes: the outer IPM iterations and the inner variational updates required to solve each Newton system. The proposed framework addresses this nested cost by using sequence-learning models that exploit early solver-generated trajectories. These early iterations provide a compact behavioral signature of the current optimization path, capturing the dominant descent direction and local response of the solver. In this section, we focus on the inner \text{(C)}VQLS level, where the early parameter-update prefix is used to predict a rich initialization for the remaining variational search. The outer IPM-level projection is described in Section~\ref{sec:learning_qipm}.

\begin{figure}[!t]
    \centering
    \includegraphics[width=0.9\columnwidth]{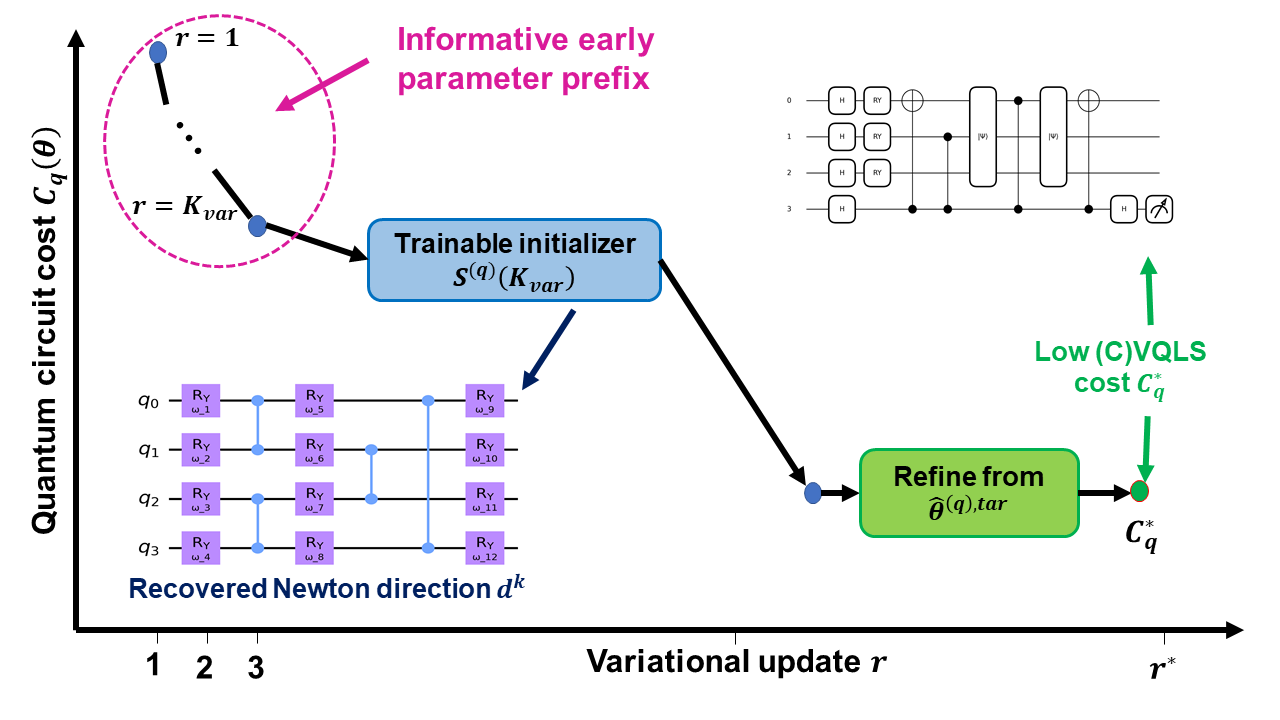}
    \caption{Trainable \text{(C)}VQLS initialization from an early parameter prefix.}
    \label{fig:cvqls_trainable_initialization}
\end{figure}

Fig.~\ref{fig:cvqls_trainable_initialization} illustrates the concept: after $K_{\mathrm{var}}$ variational updates, the parameter prefix feeds the trainable model, which predicts a continuation point; the solver then proceeds from that point.
\subsection{Information Gain from Early \text{(C)}VQLS Updates}
\label{subsec:early_information_gain}

The early part of a \text{(C)}VQLS trajectory contains information about the local optimization landscape of the current IPM-generated linear system. Let $\boldsymbol{\theta}{m,k}^{(q),\mathrm{tar}}$ denote the best-cost parameter vector along a baseline trajectory. For a prefix of length $K{\mathrm{var}}$, the minimum achievable prediction error under squared loss is:
{\small
\begin{equation}
\mathcal{E}_{q}^{\star}(K_{\mathrm{var}})
=
\frac{1}{p_q}
\mathbb{E}
\left[
\left\|
\boldsymbol{\theta}_{m,k}^{(q),\mathrm{tar}}
-
\mathbb{E}
\left[
\boldsymbol{\theta}_{m,k}^{(q),\mathrm{tar}}
\mid
\mathbf{S}_{m,k}^{(q)}(K_{\mathrm{var}})
\right]
\right\|_2^2
\right].
\label{eq:conditional_prediction_error}
\end{equation}}
The conditional expectation is the optimal predictor from the observed prefix. Thus, $\mathcal{E}{q}^{\star}(K{\mathrm{var}})$ measures the uncertainty remaining about the best-cost parameter vector after the first $K_{\mathrm{var}}$ updates.

For the optimal conditional predictor, a longer prefix contains all the information available in a shorter prefix, and therefore

{\small
\begin{equation}
\mathcal{E}_{q}^{\star}(K_{\mathrm{var}}+1)
\leq
\mathcal{E}_{q}^{\star}(K_{\mathrm{var}}).
\label{eq:error_monotonicity}
\end{equation}}
However, the useful information gained by one additional variational update is measured by the marginal reduction
{\small
\begin{equation}
\Delta \mathcal{I}_{q}(K_{\mathrm{var}})
=
\mathcal{E}_{q}^{\star}(K_{\mathrm{var}})
-
\mathcal{E}_{q}^{\star}(K_{\mathrm{var}}+1).
\label{eq:marginal_information_gain}
\end{equation}}
Early updates generally provide greater information gain by revealing the dominant descent behavior, whereas later updates mainly refine an already low-cost region.

This interpretation is consistent with the variational-parameter dynamics. The update increment
{\small
\begin{equation}
\Delta \boldsymbol{\theta}_{m,k,r}^{(q)}
=
\boldsymbol{\theta}_{m,k,r+1}^{(q)}
-
\boldsymbol{\theta}_{m,k,r}^{(q)}
\label{eq:theta_increment}
\end{equation}}
reflects the local response of the ansatz, cost function, and optimizer to the system $\mathbf{A}{m,k}\mathbf{d}{m,k}=\mathbf{b}{m,k}$. It is related to the local gradient for gradient-based methods and represents the accepted search direction for gradient-free methods. Hence, $\mathbf{S}{m,k}^{(q)}(K_{\mathrm{var}})$ provides a compact dynamic signature of the current \text{(C)}VQLS solve.

The prefix should capture most of the useful information without requiring many updates. Since observing the prefix costs approximately $\tau_q K_{\mathrm{var}}$, its length is selected using prediction error, final \text{(C)}VQLS cost, remaining updates, and prefix length. The trainable model approximates the conditional predictor in \eqref{eq:lstm_conditional_predictor} as
{\small
\begin{equation}
f_{\varphi_q}^{\mathrm{var}}
\left(
\mathbf{S}_{m,k}^{(q)}(K_{\mathrm{var}})
\right)
\approx
\mathbb{E}
\left[
\boldsymbol{\theta}_{m,k}^{(q),\mathrm{tar}}
\mid
\mathbf{S}_{m,k}^{(q)}(K_{\mathrm{var}})
\right].
\label{eq:lstm_conditional_predictor}
\end{equation}}
Accordingly, the early trajectory is used to infer a better continuation point and avoid low-information refinement updates.

\begin{remark}[Gradient content of the prefix]
\label{rem:gradient_content}
The dynamic-signature interpretation of the prefix can be made precise for gradient-based \text{(C)}VQLS updates. Under gradient descent or momentum-based rules, the update increment at step $r$ satisfies
\begin{equation}
\Delta \boldsymbol{\theta}_{m,k,r}^{(q)}
\approx
-\eta_r \nabla_{\boldsymbol{\theta}} C_q\!\left(\boldsymbol{\theta}_{m,k,r}^{(q)};\mathbf{A}_{m,k},\mathbf{b}_{m,k}\right),
\label{eq:grad_content}
\end{equation}
so the first increment ($r=0$) encodes the gradient direction at initialization, and differences between successive increments encode approximate curvature information via
\begin{equation}
\frac{\Delta\boldsymbol{\theta}_{m,k,r}^{(q)}}{\eta_r}
-
\frac{\Delta\boldsymbol{\theta}_{m,k,r-1}^{(q)}}{\eta_{r-1}}
\approx
-\nabla^2_{\boldsymbol{\theta}} C_q\!\left(\boldsymbol{\theta}_{m,k,r-1}^{(q)}\right)
\Delta\boldsymbol{\theta}_{m,k,r-1}^{(q)}.
\label{eq:curvature_content}
\end{equation}
For gradient-free updates (e.g., SPSA or COBYLA), the accepted step direction still encodes the local cost response to parameter perturbations, which is a finite-difference approximation to the gradient. In both cases, the prefix $\mathbf{S}_{m,k}^{(q)}(K_{\mathrm{var}})$ constitutes a compact dynamic signature of the local cost landscape induced by the linear system $(\mathbf{A}_{m,k},\mathbf{b}_{m,k})$, which is the theoretical basis for using it as the conditioning input to the trainable model.
\end{remark}
\subsection{Trajectory-Based Target Definition}

Let $m\in\{1,\dots,M\}$ index OPF operating scenarios generated by perturbing the base loading condition. For each scenario, IPM produces a sequence of Newton systems
{\small
\begin{equation}
\mathbf{A}_{m,k}\mathbf{d}_{m,k}=\mathbf{b}_{m,k},
\label{eq:mk_linear_system}
\end{equation}}
where $k$ is the IPM iteration index. Each system in \eqref{eq:mk_linear_system} is then solved using the corresponding baseline quantum linear solver.

For solver $q\in\{\mathrm{VQLS},\mathrm{CVQLS}\}$, the baseline variational optimization generates the parameter sequence
{\small
\begin{equation}
\Theta_{m,k}^{(q)}
=
\left[
\boldsymbol{\theta}_{m,k,0}^{(q)},
\boldsymbol{\theta}_{m,k,1}^{(q)},
\dots,
\boldsymbol{\theta}_{m,k,R_{m,k}^{(q)}}^{(q)}
\right],
\label{eq:theta_trajectory}
\end{equation}}
where $\boldsymbol{\theta}_{m,k,r}^{(q)}\in\mathbb{R}^{p_q}$ is the parameter vector at variational update $r$, and $R_{m,k}^{(q)}$ is the total number of baseline variational updates.

The supervised target is selected as the best-cost parameter vector observed along the baseline sequence:
{\small
\begin{equation}
r_{m,k}^{(q),\mathrm{tar}}
=
\arg\min_{0\le r\le R_{m,k}^{(q)}}
C_q\!\left(\boldsymbol{\theta}_{m,k,r}^{(q)}\right),
\label{eq:best_theta_index}
\end{equation}}
{\small
\begin{equation}
\boldsymbol{\theta}_{m,k}^{(q),\mathrm{tar}}
=
\boldsymbol{\theta}_{m,k,r_{m,k}^{(q),\mathrm{tar}}}^{(q)} .
\label{eq:best_theta}
\end{equation}}
Thus, the learning target is not the exact linear-system solution and is not defined by a direct one-shot estimate of the final answer. Instead, it is chosen from the solver-generated trajectory itself, consistent with the learner's role as an embedded trajectory-based projection mechanism.

\begin{remark}[Optimality of the trajectory target]
\label{rem:target_optimality}
The choice of $\boldsymbol{\theta}_{m,k}^{(q),\mathrm{tar}}$ as the best-cost point along the baseline trajectory provides the following guarantee: if the LSTM prediction is exact, i.e., $\hat{\boldsymbol{\theta}}_{m,k}^{(q),\mathrm{tar}} = \boldsymbol{\theta}_{m,k}^{(q),\mathrm{tar}}$, then initializing the variational optimization from the predicted point achieves
\begin{equation}
C_q\!\left(\hat{\boldsymbol{\theta}}_{m,k}^{(q),\mathrm{tar}}\right)
\leq
C_q\!\left(\boldsymbol{\theta}_{m,k,r}^{(q)}\right),
\quad \forall\, r \in \{0,\dots,R_{m,k}^{(q)}\},
\label{eq:target_optimality}
\end{equation}
using at most $K_{\mathrm{var}}$ variational updates instead of $R_{m,k}^{(q)}$. Moreover, if the predicted point already satisfies the solver tolerance $\epsilon_q$, no further updates are required and the total update count is exactly $K_{\mathrm{var}}$. When the prediction is imperfect, the variational optimization may continue from $\hat{\boldsymbol{\theta}}_{m,k}^{(q),\mathrm{tar}}$, but the proximity of the initialization to the low-cost region reduces the number of additional updates required relative to a generic initialization at $\boldsymbol{\theta}_{m,k,0}^{(q)}$.
\end{remark}

\subsection{Parameter-Sequence Initialization Model}

For each baseline sequence, the first $K_{\mathrm{var}}$ parameter vectors form the input sequence
{\small
\begin{equation}
\mathbf{S}_{m,k}^{(q)}(K_{\mathrm{var}})
=
\left[
\boldsymbol{\theta}_{m,k,0}^{(q)},
\boldsymbol{\theta}_{m,k,1}^{(q)},
\dots,
\boldsymbol{\theta}_{m,k,K_{\mathrm{var}}-1}^{(q)}
\right].
\label{eq:prefix_sequence}
\end{equation}}
The trainable initialization model predicts the target parameter vector as
{\small
\begin{equation}
\hat{\boldsymbol{\theta}}_{m,k}^{(q),\mathrm{tar}}
=
f_{\varphi_q}^{\mathrm{var}}
\!\left(
\mathbf{S}_{m,k}^{(q)}(K_{\mathrm{var}})
\right),
\label{eq:lstm_predictor}
\end{equation}}
where $\varphi_q$ denotes the model parameters for solver $q$. The model is trained by minimizing
{\small
\begin{equation}
\mathcal{L}_{\mathrm{var}}^{(q)}(\varphi_q)
=
\frac{1}{|\mathcal{D}_{\mathrm{var,tr}}|}
\sum_{(m,k)\in\mathcal{D}_{\mathrm{var,tr}}}
\frac{
\left\|
f_{\varphi_q}^{\mathrm{var}}
\!\left(
\mathbf{S}_{m,k}^{(q)}(K_{\mathrm{var}})
\right)
-
\boldsymbol{\theta}_{m,k}^{(q),\mathrm{tar}}
\right\|_2^2
}{p_q}.
\label{eq:lstm_loss}
\end{equation}}

\begin{remark}[LSTM as a universal sequence-to-parameter approximator]
\label{rem:lstm_universality}
The mapping $\mathbf{S}_{m,k}^{(q)}(K_{\mathrm{var}}) \mapsto \boldsymbol{\theta}_{m,k}^{(q),\mathrm{tar}}$ is an output functional of the discrete-time nonlinear dynamical system
\begin{equation}
\boldsymbol{\theta}_{m,k,r+1}^{(q)}
=
\Phi_q\!\left(\boldsymbol{\theta}_{m,k,r}^{(q)};\,\mathbf{A}_{m,k},\mathbf{b}_{m,k}\right),
\label{eq:vqls_dynamics}
\end{equation}
where $\Phi_q$ is the \text{(C)}VQLS parameter update rule. Under mild stability conditions (specifically, the fading-memory property, which requires that the influence of early inputs on the system state decays with time), output functionals of such discrete-time nonlinear systems can be approximated to arbitrary accuracy by LSTMs with sufficient hidden dimension~\cite{schafer2007recurrent}. The LSTM in \eqref{eq:lstm_predictor} therefore provides a universal function class for approximating the conditional mean predictor in \eqref{eq:lstm_conditional_predictor} from the observed prefix.
\end{remark}

When angular parameters are used, they are represented by sine--cosine pairs and normalized consistently during training and inference. Separate models are used when the solver type, ansatz, or parameter dimension differs.

For a new IPM-generated system $\mathbf{A}\mathbf{d}=\mathbf{b}$, the \text{(C)}VQLS optimizer is first run for $K_{\mathrm{var}}$ updates. The resulting prefix is then passed to the trained model, which predicts
{\small
\begin{equation}
\hat{\boldsymbol{\theta}}^{(q),\mathrm{tar}}
=
f_{\varphi_q}^{\mathrm{var}}
\!\left(
\mathbf{S}^{(q)}(K_{\mathrm{var}})
\right),
\label{eq:test_prediction}
\end{equation}}
and the remaining optimization is initialized as
{\small
\begin{equation}
\boldsymbol{\theta}_{K_{\mathrm{var}}}^{(q)}
\leftarrow
\hat{\boldsymbol{\theta}}^{(q),\mathrm{tar}}.
\label{eq:initialization_replace}
\end{equation}}
If the initialized point already satisfies the solver tolerance, the solve is terminated; otherwise, the variational optimization continues from the predicted parameter vector. Therefore, the model uses the early sequence to improve the remaining \text{(C)}VQLS search rather than replacing it with direct parameter regression.

\subsection{Prefix-Length Selection}

The prefix length $K_{\mathrm{var}}$ controls how many early \text{(C)}VQLS updates are observed before the trainable initializer is invoked. A short prefix may not contain enough information about the parameter search, while a long prefix reduces the benefit of initialization. Therefore, $K_{\mathrm{var}}$ is selected from a candidate set $\mathcal{K}_{\mathrm{var}}$ using validation performance.

For each $K_{\mathrm{var}}\in\mathcal{K}_{\mathrm{var}}$, the trained model is evaluated using the parameter-prediction error $E_{\theta}(K_{\mathrm{var}})$, the final \text{(C)}VQLS cost $\bar{C}_{\mathrm{final}}(K_{\mathrm{var}})$, and the average number of remaining variational updates $\bar{N}_{\mathrm{rem}}(K_{\mathrm{var}})$. Each term is normalized before forming the validation score:
{\small
\begin{equation}
\tilde{Y}(K_{\mathrm{var}})
=
\frac{
Y(K_{\mathrm{var}})
}{
\max_{K\in\mathcal{K}_{\mathrm{var}}}Y(K)
+\epsilon
},
\label{eq:normalized_metric}
\end{equation}}
where $Y$ denotes $E_{\theta}$, $\bar{C}_{\mathrm{final}}$, or $\bar{N}_{\mathrm{rem}}$, and $\epsilon>0$ avoids division by zero. The validation score is
{\small
\begin{equation}
\begin{aligned}
\mathcal{J}_{\mathrm{pre}}(K_{\mathrm{var}})
=&\;
\rho_1 \tilde{E}_{\theta}(K_{\mathrm{var}})
+\rho_2 \tilde{C}_{\mathrm{final}}(K_{\mathrm{var}}) \\
&+
\rho_3 \tilde{N}_{\mathrm{rem}}(K_{\mathrm{var}})
+\rho_4 \frac{K_{\mathrm{var}}}{K_{\mathrm{var}}^{\max}} ,
\end{aligned}
\label{eq:prefix_score}
\end{equation}}
where $\rho_j\ge 0$, $j=1,\dots,4$, are weighting factors. The selected prefix length is
{\small
\begin{equation}
K_{\mathrm{var}}^{\star}
=
\min\arg\min_{K_{\mathrm{var}}\in\mathcal{K}_{\mathrm{var}}}
\mathcal{J}_{\mathrm{pre}}(K_{\mathrm{var}}),
\label{eq:prefix_best}
\end{equation}}
where the outer minimum selects the shortest prefix among candidates with comparable validation performance.

\section{Trajectory-Informed IPM Central Path Projection} 
\label{sec:learning_qipm}

The outer-level model addresses the number of IPM iterations. The premise is that the first few primal-dual IPM iterates provide a compact representation of the local central path direction, encoding the evolution of primal and dual variables, slack variables, and the barrier parameter. An LSTM uses this prefix to project a later central path state; the projected point is then restored and refined by IPM.

\begin{figure}[!t]
    \centering
    \includegraphics[width=0.9\columnwidth]{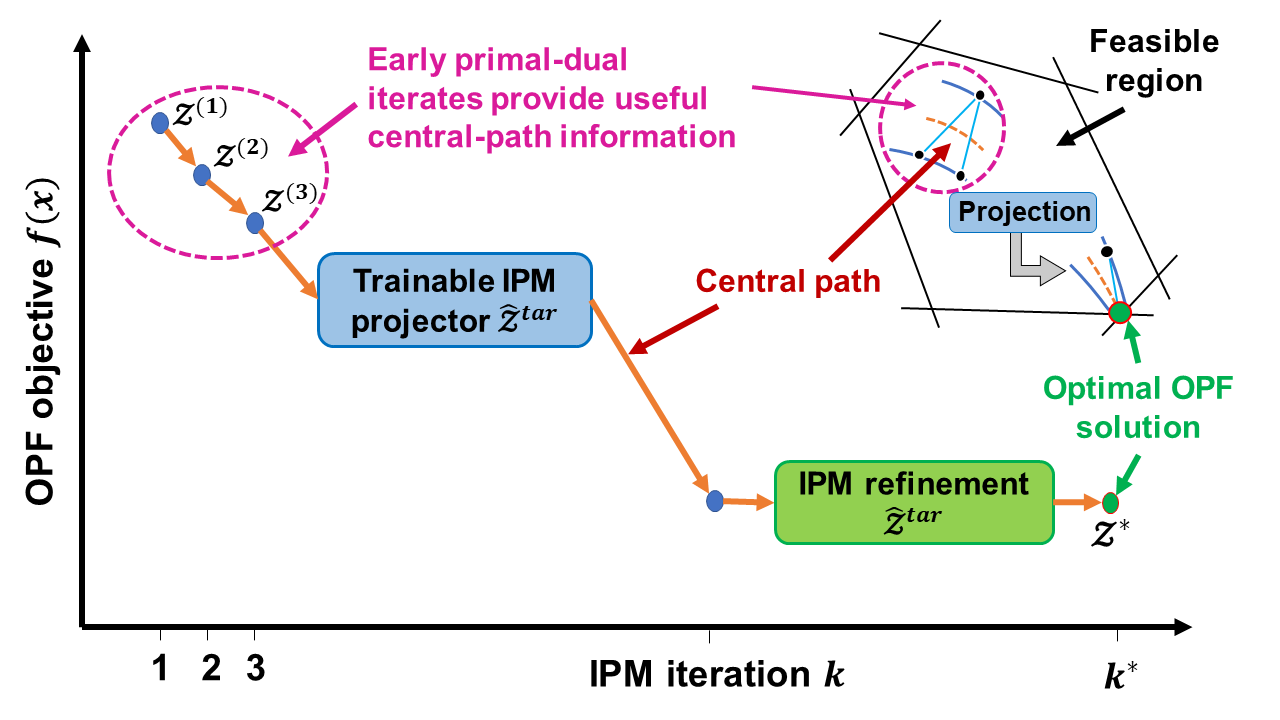}
    \caption{Trajectory-informed IPM central path projection.}
    \label{fig:ipm_central_path_projection}
\end{figure}

Fig.~\ref{fig:ipm_central_path_projection} illustrates this mechanism.
\subsection{Central Path Information in Early IPM Iterates}
\label{subsec:ipm_early_information}

The proposed OPF-level projection is based on the fact that early IPM iterates contain information about the local central path geometry of the current OPF instance. At iteration $k$, the primal-dual IPM computes a Newton correction from
{\small
\begin{equation}
\mathbf{J}(\mathbf{z}_{m,k})\Delta \mathbf{z}_{m,k}
=
-\mathbf{F}(\mathbf{z}_{m,k}),
\label{eq:ipm_newton_local_info}
\end{equation}}
and updates the state according to
{\small
\begin{equation}
\mathbf{z}_{m,k+1}
=
\mathbf{z}_{m,k}
+
\alpha_{m,k}\Delta \mathbf{z}_{m,k},
\label{eq:ipm_state_update_info}
\end{equation}}
where $\alpha_{m,k}$ is the step size. Therefore, the observed state increments $\boldsymbol{\delta}_{m,k}^{\mathrm{ipm}}$ in~\eqref{eq:observed_ipm_increment} encode the local response of the KKT residual, active-limit behavior, slack variables, dual variables, and barrier parameter for the current scenario.
{\small
\begin{equation}
\boldsymbol{\delta}_{m,k}^{\mathrm{ipm}}
=
\mathbf{z}_{m,k+1}
-
\mathbf{z}_{m,k}.
\label{eq:observed_ipm_increment}
\end{equation}}
Thus, the prefix is not only a collection of early states; it is a short observation of the Newton dynamics that govern the movement toward the central path.

Let $\mathbf{z}_{m}^{\mathrm{tar}}$ denote a later near-converged IPM point. For a prefix of length $K_{\mathrm{ipm}}$, define the minimum achievable prediction error under squared loss as
{\small
\begin{equation}
\mathcal{E}_{\mathrm{ipm}}^{\star}(K_{\mathrm{ipm}})
=
\mathbb{E}
\left[
\left\|
\mathbf{z}_{m}^{\mathrm{tar}}
-
\mathbb{E}
\left[
\mathbf{z}_{m}^{\mathrm{tar}}
\mid
\mathbf{P}_{m}(K_{\mathrm{ipm}})
\right]
\right\|_{\boldsymbol{\Omega}_{\mathrm{ipm}}}^2
\right],
\label{eq:ipm_conditional_error}
\end{equation}}
where $\|\mathbf{a}\|_{\boldsymbol{\Omega}_{\mathrm{ipm}}}^2 = \mathbf{a}^{\top}\boldsymbol{\Omega}_{\mathrm{ipm}}\mathbf{a}$ is a weighted norm used to account for different scales of the primal, dual, slack, and barrier components. The conditional expectation is the optimal predictor of the later central path point given the observed prefix. Hence, $\mathcal{E}_{\mathrm{ipm}}^{\star}(K_{\mathrm{ipm}})$ measures the remaining uncertainty about the later IPM state after observing the first $K_{\mathrm{ipm}}$ iterates.

Since a longer prefix contains at least the same information as a shorter prefix,
{\small
\begin{equation}
\mathcal{E}_{\mathrm{ipm}}^{\star}(K_{\mathrm{ipm}}+1)
\leq
\mathcal{E}_{\mathrm{ipm}}^{\star}(K_{\mathrm{ipm}}).
\label{eq:ipm_error_monotonicity}
\end{equation}}
The marginal value of one additional IPM iteration can therefore be written as
{\small
\begin{equation}
\Delta \mathcal{I}_{\mathrm{ipm}}(K_{\mathrm{ipm}})
=
\mathcal{E}_{\mathrm{ipm}}^{\star}(K_{\mathrm{ipm}})
-
\mathcal{E}_{\mathrm{ipm}}^{\star}(K_{\mathrm{ipm}}+1).
\label{eq:ipm_marginal_information_gain}
\end{equation}}
Early IPM iterations typically provide large information gain because they reveal the dominant movement of the primal-dual variables and the active-constraint pattern. After the trajectory enters the local convergence region, additional iterations mainly reduce the residual and complementarity error. These later iterations remain important for final accuracy, but their marginal information for predicting the central path direction is smaller.

This motivates using a short prefix to project the trajectory to a later IPM state. The learned model approximates
{\small
\begin{equation}
g_{\psi}^{\mathrm{ipm}}
\left(
\mathbf{P}_{m}(K_{\mathrm{ipm}})
\right)
\approx
\mathbb{E}
\left[
\mathbf{z}_{m}^{\mathrm{tar}}
\mid
\mathbf{P}_{m}(K_{\mathrm{ipm}})
\right].
\label{eq:ipm_conditional_predictor}
\end{equation}}
The predicted point is restored and refined by IPM, reducing low-information intermediate iterations while retaining the solver's feasibility and optimality enforcement.
For OPF scenario $m$, the IPM state at iteration $k$ is defined as
{\small
\begin{equation}
\mathbf{z}_{m,k}
=
\left[
\mathbf{x}_{m,k}^{\top},
\boldsymbol{\lambda}_{m,k}^{\top},
\boldsymbol{\nu}_{m,k}^{\top},
\mathbf{s}_{m,k}^{\top},
\mu_{m,k}
\right]^{\top},
\label{eq:ipm_state_vector}
\end{equation}}
where $\mathbf{x}_{m,k}$ is the OPF primal-variable vector, $\boldsymbol{\lambda}_{m,k}$ and $\boldsymbol{\nu}_{m,k}$ are the dual variables associated with equality and inequality constraints, $\mathbf{s}_{m,k}$ is the slack-variable vector, and $\mu_{m,k}$ is the barrier parameter. One IPM solve therefore generates the central path sequence
{\small
\begin{equation}
\mathcal{Z}_{m}
=
\left[
\mathbf{z}_{m,0},
\mathbf{z}_{m,1},
\dots,
\mathbf{z}_{m,T_m}
\right],
\label{eq:ipm_trajectory}
\end{equation}}
where $T_m$ is the number of IPM iterations required for convergence.

The first $K_{\mathrm{ipm}}$ iterates form the IPM prefix
{\small
\begin{equation}
\mathbf{P}_{m}(K_{\mathrm{ipm}})
=
\left[
\mathbf{z}_{m,0},
\mathbf{z}_{m,1},
\dots,
\mathbf{z}_{m,K_{\mathrm{ipm}}-1}
\right].
\label{eq:opf_prefix}
\end{equation}}

\subsection{LSTM-Based Central Path Projection}
\label{subsec:lstm_opf}

The IPM-level LSTM model $g_{\psi}^{\mathrm{ipm}}(\cdot)$ predicts a later primal-dual point from the prefix sequence:
{\small
\begin{equation}
\hat{\mathbf{z}}_{m}^{\mathrm{tar}}
=
g_{\psi}^{\mathrm{ipm}}
\!\left(
\mathbf{P}_{m}(K_{\mathrm{ipm}})
\right),
\label{eq:opf_lstm_predictor}
\end{equation}}
where $\psi$ denotes the trainable parameters of the IPM-level LSTM. Let $\boldsymbol{\eta}_{m,k}$ and $\boldsymbol{\chi}_{m,k}$ denote the hidden and cell states. The recurrent update is
{\small
\begin{equation}
\left(\boldsymbol{\eta}_{m,k},\boldsymbol{\chi}_{m,k}\right)
=
\mathrm{LSTM}_{\psi}
\left(
\mathbf{z}_{m,k},
\boldsymbol{\eta}_{m,k-1},
\boldsymbol{\chi}_{m,k-1}
\right),
\label{eq:opf_lstm_recurrence}
\end{equation}}
and the final hidden state is mapped to the predicted IPM point by
{\small
\begin{equation}
\hat{\mathbf{z}}_{m}^{\mathrm{tar}}
=
\mathbf{W}_{\mathrm{ipm}}
\boldsymbol{\eta}_{m,K_{\mathrm{ipm}}-1}
+
\boldsymbol{\beta}_{\mathrm{ipm}},
\label{eq:opf_lstm_readout}
\end{equation}}
where $\mathbf{W}_{\mathrm{ipm}}$ and $\boldsymbol{\beta}_{\mathrm{ipm}}$ are the readout weight matrix and bias vector.

The target point is selected as a later near-converged IPM state,
{\small
\begin{equation}
\mathbf{z}_{m}^{\mathrm{tar}}=\mathbf{z}_{m,t^\star},
\label{eq:opf_target}
\end{equation}}
where $t^\star$ denotes a later iteration close to convergence. In the simplest implementation, $t^\star=T_m$ and the final IPM iterate is used. Alternatively, $t^\star$ may be chosen as the best feasible iterate based on the KKT residual, objective value, and constraint violation.

\subsection{Training Objective}
\label{subsec:opf_lstm_data}

The OPF-level training set is formed from disjoint operating scenarios as
{\small
\begin{equation}
\mathcal{D}_{\mathrm{ipm,tr}}
=
\left\{
\left(
\mathbf{P}_{m}(K_{\mathrm{ipm}}),
\mathbf{z}_{m}^{\mathrm{tar}}
\right)
\right\}_{m\in\mathcal{I}_{\mathrm{tr}}},
\label{eq:opf_training_set}
\end{equation}}
where $\mathcal{I}_{\mathrm{tr}}$ is the index set of training scenarios. Validation and test sets, denoted by $\mathcal{D}_{\mathrm{ipm,val}}$ and $\mathcal{D}_{\mathrm{ipm,te}}$, are formed from separate OPF scenarios.

The model is trained by minimizing
{\small
\begin{equation}
\mathcal{L}_{\mathrm{ipm}}(\psi)
=
\frac{1}{|\mathcal{D}_{\mathrm{ipm,tr}}|}
\sum_{m\in\mathcal{I}_{\mathrm{tr}}}
\left\|
g_{\psi}^{\mathrm{ipm}}
\left(
\mathbf{P}_{m}(K_{\mathrm{ipm}})
\right)
-
\mathbf{z}_{m}^{\mathrm{tar}}
\right\|_2^2.
\label{eq:opf_lstm_loss}
\end{equation}}
If the components of $\mathbf{z}_{m,k}$ have different numerical scales, the IPM state vectors are normalized during training, and the same normalization is applied during inference.
\subsection{IPM Prefix-Length Selection}
\label{subsec:ipm_prefix_selection}

The prefix length $K_{\mathrm{ipm}}$ is selected using validation trajectories based on IPM-state prediction accuracy and the correction iterations required after restoration. Although longer prefixes provide more central-path information, improvements beyond the first few iterations are marginal. In particular, using more than three iterates did not meaningfully reduce prediction error or correction effort, while increasing the iterations completed before projection. Therefore, $K_{\mathrm{ipm}}=3$ was selected as the shortest prefix that provided stable predictions and reliable IPM refinement, and it was used for all test systems.

\subsection{Restoration of Predicted IPM Point}
\label{subsec:opf_restoration}

The predicted IPM point may not satisfy the numerical requirements of the primal-dual IPM. Therefore, it is restored before being used by the solver:
{\small
\begin{equation}
\tilde{\mathbf{z}}_{m}^{\mathrm{tar}}
=
\mathcal{R}
\left(
\hat{\mathbf{z}}_{m}^{\mathrm{tar}}
\right).
\label{eq:opf_restored_state}
\end{equation}}
For $\mathbf{z}=[\mathbf{x}^{\top},\boldsymbol{\lambda}^{\top},\boldsymbol{\nu}^{\top},\mathbf{s}^{\top},\mu]^{\top}$, the restoration operator enforces the basic admissibility conditions required for IPM continuation:
{\small
\begin{equation}
\begin{aligned}
\tilde{\mathbf{x}} &= \Pi_{\mathcal{X}}(\hat{\mathbf{x}}), &
\tilde{\mathbf{s}} &= \max\{\hat{\mathbf{s}},\epsilon_s\mathbf{1}\}, \\
\tilde{\boldsymbol{\nu}} &= \max\{\hat{\boldsymbol{\nu}},\epsilon_\nu\mathbf{1}\}, &
\tilde{\mu} &= \frac{\tilde{\mathbf{s}}^{\top}\tilde{\boldsymbol{\nu}}}{n_h}.
\end{aligned}
\label{eq:restore_compact}
\end{equation}}
Here, $\Pi_{\mathcal{X}}(\cdot)$ denotes projection onto the primal-variable bounds, $\epsilon_s$ and $\epsilon_\nu$ are small positive constants, and $n_h$ is the number of inequality constraints.

\begin{proposition}[Unconditional admissibility of the restored point]
\label{prop:restoration}
For any predicted point $\hat{\mathbf{z}}^{\mathrm{tar}}$, the operator $\mathcal{R}(\cdot)$ defined in \eqref{eq:restore_compact} produces a point $\tilde{\mathbf{z}} = \mathcal{R}(\hat{\mathbf{z}}^{\mathrm{tar}})$ satisfying the following properties unconditionally, irrespective of the prediction quality of $\hat{\mathbf{z}}^{\mathrm{tar}}$:
\begin{enumerate}
\item[(i)] \textbf{Strict interior feasibility}: $\tilde{s}_i \geq \epsilon_s > 0$ and $\tilde{\nu}_i \geq \epsilon_\nu > 0$ for all $i = 1,\dots,n_h$.
\item[(ii)] \textbf{Valid barrier parameter}: $\tilde{\mu} \geq \epsilon_s \epsilon_\nu > 0$.
\item[(iii)] \textbf{Centrality ratio bound}: each complementarity ratio $\gamma_i := \tilde{s}_i\tilde{\nu}_i/\tilde{\mu}$ satisfies
\begin{equation}
0 \;<\;
\underline{\gamma}
:=
\frac{n_h \epsilon_s \epsilon_\nu}{\tilde{\mathbf{s}}^{\top}\tilde{\boldsymbol{\nu}}}
\;\leq\;
\gamma_i
\;\leq\;
\bar{\gamma}
:=
\frac{n_h \max_j(\tilde{s}_j\tilde{\nu}_j)}{\tilde{\mathbf{s}}^{\top}\tilde{\boldsymbol{\nu}}}
\;\leq\; n_h.
\label{eq:centrality_bounds}
\end{equation}
\end{enumerate}
In particular, $\tilde{\mathbf{z}}$ constitutes a valid strictly interior starting point for IPM continuation.
\end{proposition}

\begin{IEEEproof}
\textit{(i)} By construction: $\tilde{s}_i = \max\{\hat{s}_i,\epsilon_s\} \geq \epsilon_s > 0$ and $\tilde{\nu}_i = \max\{\hat{\nu}_i,\epsilon_\nu\} \geq \epsilon_\nu > 0$.
\textit{(ii)} From (i): $\tilde{\mu} = \tilde{\mathbf{s}}^{\top}\tilde{\boldsymbol{\nu}}/n_h \geq n_h\epsilon_s\epsilon_\nu/n_h = \epsilon_s\epsilon_\nu > 0$.
\textit{(iii)} Write $\gamma_i = n_h\tilde{s}_i\tilde{\nu}_i/(\tilde{\mathbf{s}}^{\top}\tilde{\boldsymbol{\nu}})$.
Lower bound: $\tilde{s}_i\tilde{\nu}_i \geq \epsilon_s\epsilon_\nu$ yields $\gamma_i \geq n_h\epsilon_s\epsilon_\nu/(\tilde{\mathbf{s}}^{\top}\tilde{\boldsymbol{\nu}}) = \underline{\gamma}$.
Upper bound: $\tilde{s}_i\tilde{\nu}_i \leq \max_j(\tilde{s}_j\tilde{\nu}_j)$ yields $\gamma_i \leq n_h\max_j(\tilde{s}_j\tilde{\nu}_j)/(\tilde{\mathbf{s}}^{\top}\tilde{\boldsymbol{\nu}}) = \bar{\gamma}$.
Finally, since $\max_j(\tilde{s}_j\tilde{\nu}_j) \leq \sum_j\tilde{s}_j\tilde{\nu}_j = \tilde{\mathbf{s}}^{\top}\tilde{\boldsymbol{\nu}}$ (all terms positive), we have $\bar{\gamma} \leq n_h$.
\end{IEEEproof}

\begin{remark}
Property (i) of Proposition~\ref{prop:restoration} guarantees that $\tilde{\mathbf{z}}$ satisfies the strict interior condition required for all primal-dual IPM iterations, regardless of how inaccurate the LSTM prediction is. Property (iii) bounds the centrality deviation from perfect complementarity ($\gamma_i = 1$ for all $i$): the upper bound $\bar{\gamma} \leq n_h$ is consistent with the standard $\infty$-neighborhood conditions used in IPM convergence analysis~\cite{wright1997primal}. Restoration thus acts as an unconditional safety layer — the subsequent IPM iterations are responsible for driving $\gamma_i \to 1$ and enforcing feasibility and optimality.
\end{remark}

\subsection{Embedding in Q-IPM Solver}
\label{subsec:opf_lstm_online}

For a new OPF scenario, IPM is first run for $K_{\mathrm{ipm}}$ iterations, producing
{\small
\begin{equation}
\mathbf{P}(K_{\mathrm{ipm}})
=
\left[
\mathbf{z}_{0},
\mathbf{z}_{1},
\dots,
\mathbf{z}_{K_{\mathrm{ipm}}-1}
\right].
\label{eq:opf_test_prefix}
\end{equation}}
The trained IPM-level LSTM predicts
{\small
\begin{equation}
\hat{\mathbf{z}}^{\mathrm{tar}}
=
g_{\psi}^{\mathrm{ipm}}
\left(
\mathbf{P}(K_{\mathrm{ipm}})
\right).
\label{eq:opf_test_prediction}
\end{equation}}
After restoration, the IPM state is reset as
{\small
\begin{equation}
\mathbf{z}_{K_{\mathrm{ipm}}}
\leftarrow
\mathcal{R}
\left(
\hat{\mathbf{z}}^{\mathrm{tar}}
\right).
\label{eq:opf_projection}
\end{equation}}
The IPM then continues from this projected point until the stopping tolerance $\epsilon_{\mathrm{IPM}}$ is satisfied.

Therefore, the proposed framework reduces computational effort at two connected levels: the \text{(C)}VQLS-level model reduces the variational effort per Newton solve, while the IPM-level model reduces the number of outer IPM iterations by projecting the central path toward a near-solution region.

\begin{remark}[Inexact Newton direction quality]
\label{rem:inexact_newton}
The \text{(C)}VQLS-recovered direction $\hat{\mathbf{d}}^{(k)}$ is generally an inexact solution to \eqref{eq:linear_system_general}. Define the relative linear-system residual
\begin{equation}
\rho_k
:=
\frac{\bigl\|\mathbf{A}^{(k)}\hat{\mathbf{d}}^{(k)} - \mathbf{b}^{(k)}\bigr\|}
{\bigl\|\mathbf{b}^{(k)}\bigr\|}.
\label{eq:inexact_residual}
\end{equation}
By standard results for inexact Newton methods applied to KKT systems~\cite{dembo1982inexact,wright1997primal}, the IPM iteration continues to reduce the KKT residual $\|\mathbf{F}(\mathbf{z}^{(k)})\|$ provided $\rho_k < 1$, and the convergence rate degrades gracefully as $\rho_k$ increases (recovering the standard local convergence rate as $\rho_k \to 0$). The \text{(C)}VQLS cost $C_q(\boldsymbol{\theta})$ is directly related to $\rho_k$: when $C_q \to 0$, the ansatz state $\mathbf{A}|d(\boldsymbol{\theta})\rangle$ becomes proportional to $|b\rangle$, which corresponds to $\rho_k \to 0$. The proposed LSTM initialization targets the low-cost parameter region identified by the early trajectory prefix, yielding small $\rho_k$ and fewer \text{(C)}VQLS updates, as reflected in the small OPF objective deviations reported in Table~\ref{tab:opf_objective_error}.
\end{remark}

Fig.~\ref{fig:proposed_nested_framework} summarizes the proposed trajectory-informed extension of the baseline Q-IPM structure. The left side shows the \text{(C)}VQLS-level initialization from the early parameter prefix $\mathbf{S}^{(q)}(K_{\mathrm{var}})$, while the right side shows the IPM-level projection from the early primal-dual prefix $\mathbf{P}(K_{\mathrm{ipm}})$.

\begin{figure}[!t]
    \centering
    \includegraphics[width=1\columnwidth]{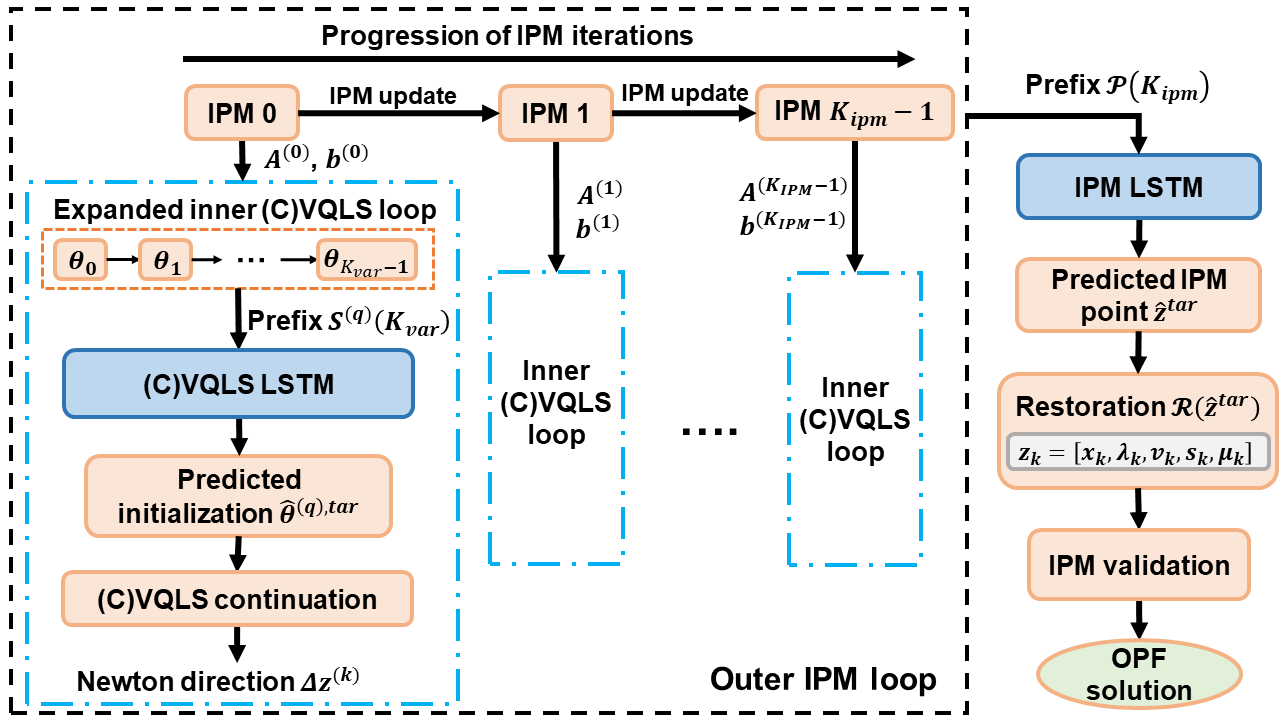}
    \caption{Nested-loop trajectory-informed Q-IPM framework.}
    \label{fig:proposed_nested_framework}
\end{figure}
\section{Case Study}

\subsection{Simulation Setting}

The proposed framework is evaluated on 2-bus, 3-bus, 5-bus, and 12-bus OPF systems. The 2-bus case is implemented on real quantum hardware through the IBMQ platform. The 3-bus and 5-bus cases are simulated using VQLS in PennyLane with Python \cite{CoherentVar}. For the 12-bus case, CVQLS is used because the embedded Newton systems are larger, making CVQLS more tractable in this setting. OPF simulations are performed using MATPOWER 8.0 \cite{zimmerman2016matpower}.

For the 3-bus and 5-bus systems, 1000 loading scenarios are generated by perturbing the nominal load within $\pm25\%$. For the 12-bus system, 2500 scenarios are used, and for the 2-bus hardware case, 100 scenarios are used. The scenarios include normal and stressed operating conditions, including cases with congestion and active operating limits. For each scenario, the OPF problem is solved using IPM, and the Newton systems along the IPM trajectory are recorded.

\subsection{Trainable \text{(C)}VQLS Implementation}
\label{subsec:lstm_impl_vqls}

The recorded Newton systems are solved using the corresponding baseline quantum linear solver. VQLS is used for the smaller systems, while CVQLS is used for the 12-bus system. The resulting variational-parameter sequences are used to train separate models. For each sequence, the first $K_{\mathrm{var}}$ variational iterates are used as the input, and the best-cost parameter vector along the baseline trajectory is used as the target. Angle-valued parameters are represented using sine and cosine components to avoid wrap-around discontinuities. The models are trained using a train/validation/test split with early stopping based on validation loss.
\subsection{Validation of Prefix-Length Selection}
\label{subsec:prefix_validation}

The prefix length $K_{\mathrm{var}}$ determines how many early \text{(C)}VQLS updates are observed before the trainable parameter prediction is applied. A very short prefix may not contain enough information, while a long prefix reduces the benefit of reducing the remaining variational search. Therefore, $K_{\mathrm{var}}$ is selected using the validation criterion in Section~\ref{sec:learning_vqls}.

The validation uses trajectories from the 2-bus, 3-bus, 5-bus, and 12-bus cases. The candidate set is
{\small
\begin{equation}
\mathcal{K}_{\mathrm{var}}=\{2,5,7,11,13\},
\end{equation}}
with $K_{\mathrm{var}}^{\max}=13$. For each candidate, the model is evaluated using the parameter-prediction error, final \text{(C)}VQLS cost, remaining variational updates, and prefix length. These quantities are normalized before forming the validation score, with equal weights.

\begin{figure}[!t]
    \centering
    \subfloat[Aggregate validation score]{%
        \includegraphics[width=0.44\columnwidth]{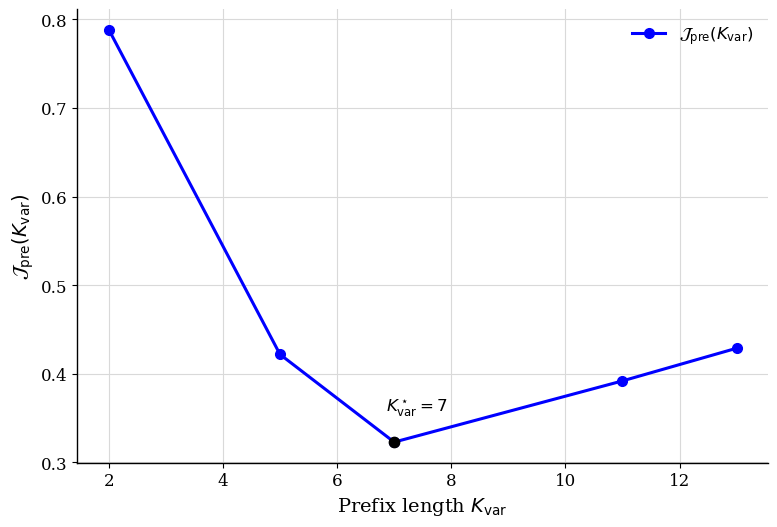}}
    \hfill
    \subfloat[Normalized score components]{%
        \includegraphics[width=0.44\columnwidth]{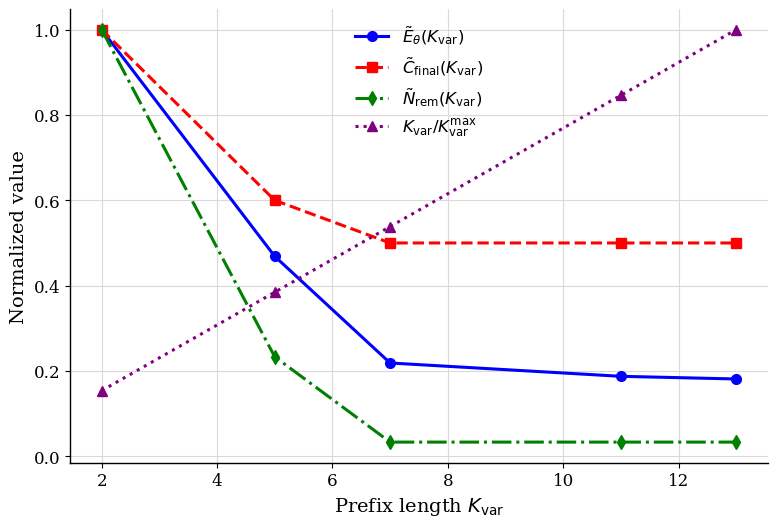}}
    \caption{Aggregate prefix-length validation across the 2-bus, 3-bus, 5-bus, and 12-bus systems.}
    \label{fig:prefix_validation}
\end{figure}

Fig.~\ref{fig:prefix_validation}(a) shows the aggregate normalized validation score. The score decreases sharply up to $K_{\mathrm{var}}=7$, indicating that the early \text{(C)}VQLS updates provide useful information for parameter prediction. Beyond this point, the improvement is marginal while the prefix-length penalty increases; hence, \(K_{\mathrm{var}}^\star=7\) is selected. Fig.~\ref{fig:prefix_validation}(b) shows the normalized score components. As $K_{\mathrm{var}}$ increases, the prediction error and remaining variational updates decrease, while the prefix-length term increases. Therefore, \(K_{\mathrm{var}}^\star=7\) provides a practical balance between prediction quality and computational effort and is used in the reported experiments.

\subsection{Real Quantum Hardware Demonstration}
\label{subsec:real_hw_demo}

The 2-bus case is used as a proof-of-concept hardware demonstration. The IPM-generated Newton systems are solved using VQLS on the IBMQ platform, and the measured variational trajectories are used to evaluate the trainable parameter initialization under hardware noise.

\begin{figure}[!t]
    \centering
    \subfloat[VQLS and L-VQLS cost on IBMQ for a representative Newton system]{%
        \includegraphics[width=0.44\columnwidth]{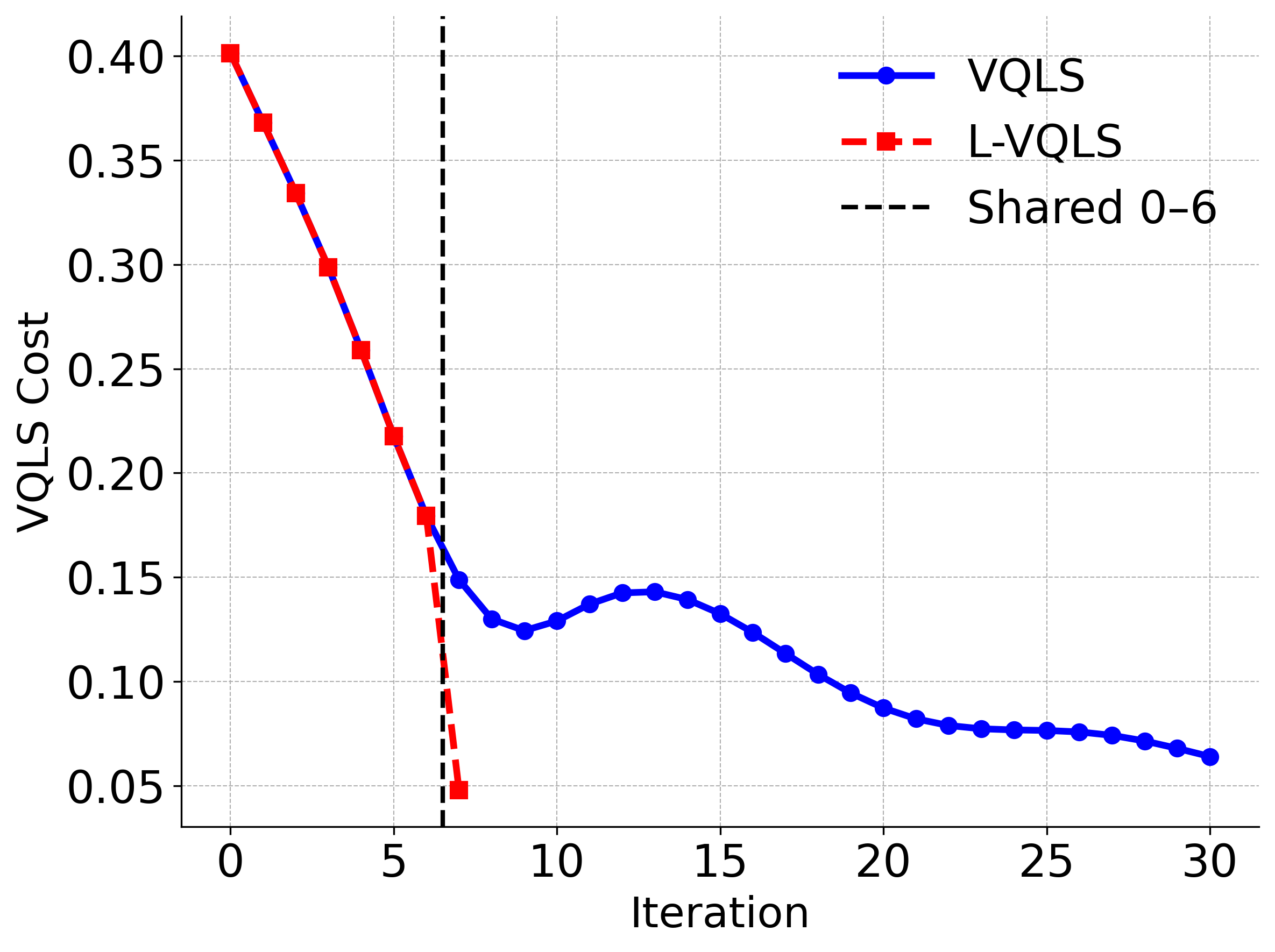}}
    \hfill
    \subfloat[OPF objective-value trajectory for classical IPM, plain QIPM, and proposed QIPM]{%
        \includegraphics[width=0.44\columnwidth]{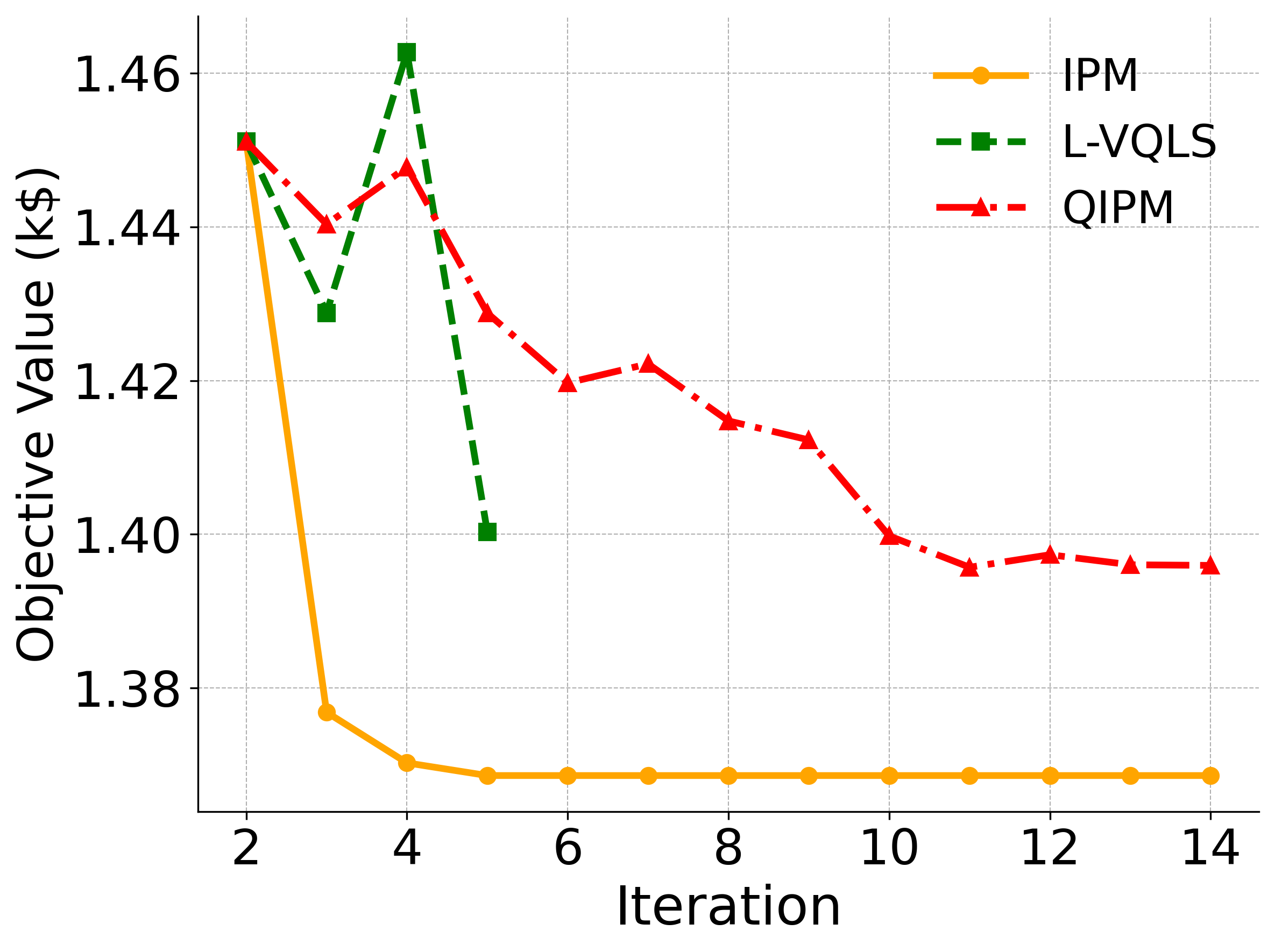}}
    \caption{Results of the 2-bus real quantum hardware study.}
    \label{fig:2bus_hardware_results}
\end{figure}

Fig.~\ref{fig:2bus_hardware_results}(a) shows a representative hardware VQLS run. The baseline and proposed methods share the same initial prefix. After the predicted parameter vector is applied, L-VQLS moves directly to a lower-cost region, while the baseline VQLS requires additional hardware-based updates.

Fig.~\ref{fig:2bus_hardware_results}(b) shows the corresponding OPF objective trajectories. The classical IPM reaches the reference objective value of approximately $1368.6$. The plain hardware-based QIPM reaches about $1395.95$ after 13 iterations, while the proposed version reaches about $1400.30$ by iteration 5. The small final difference is mainly attributed to the approximate Newton directions produced by the hardware-based VQLS solve.

\subsection{VQLS Results}
\label{subsec:learning_vqls_results}

The 3-bus and 5-bus cases evaluate the proposed trainable VQLS parameter initialization on IPM-generated linear systems with embedded dimensions up to $32\times32$. These systems remain computationally feasible for VQLS simulation while still showing the effect of repeated variational optimization inside IPM.

\begin{figure}[!t]
    \centering
    \subfloat[3-bus system, $16\times16$ matrix]{%
        \includegraphics[width=0.44\columnwidth]{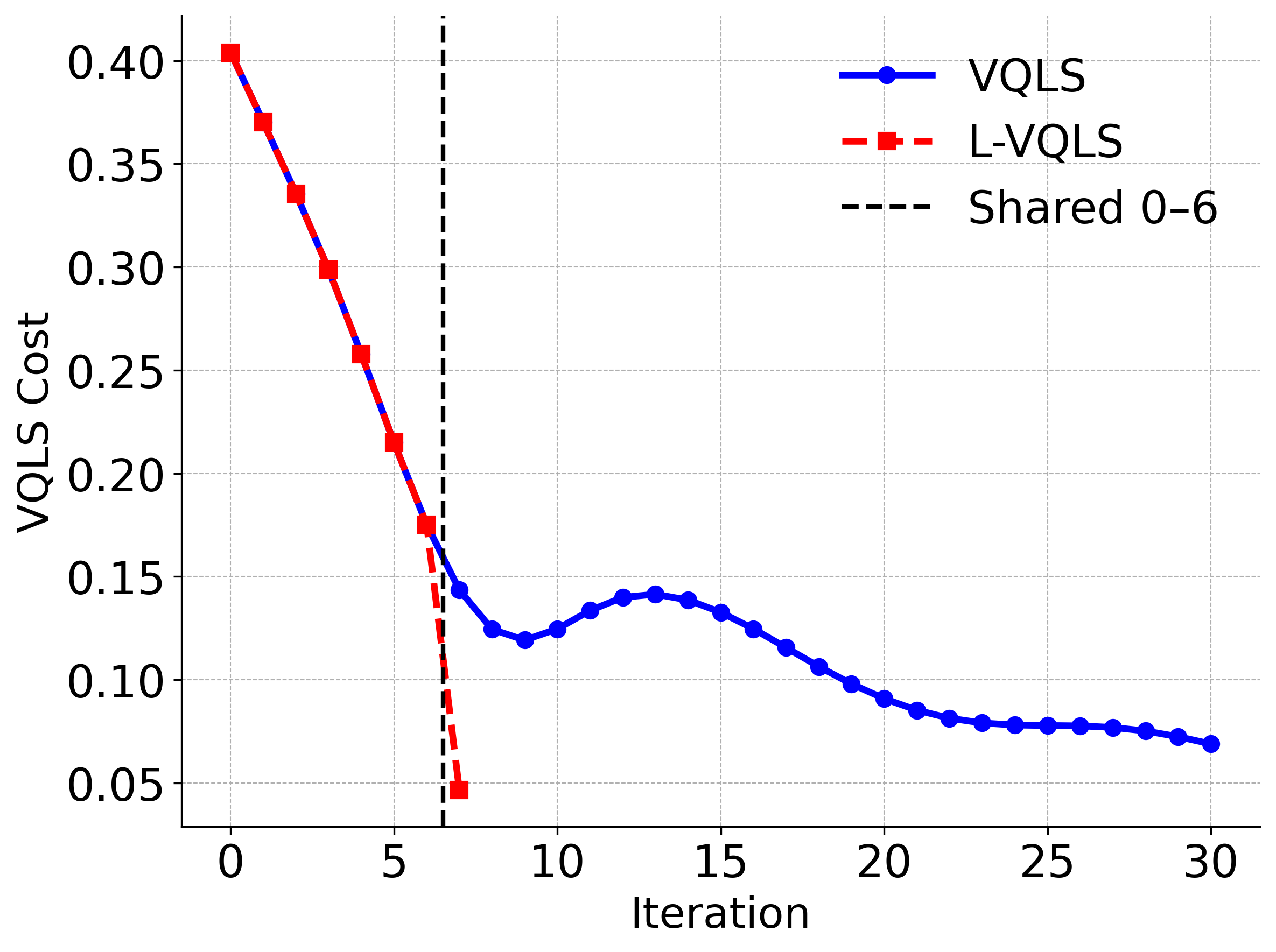}}
    \hfill
    \subfloat[5-bus system, $32\times32$ matrix]{%
        \includegraphics[width=0.44\columnwidth]{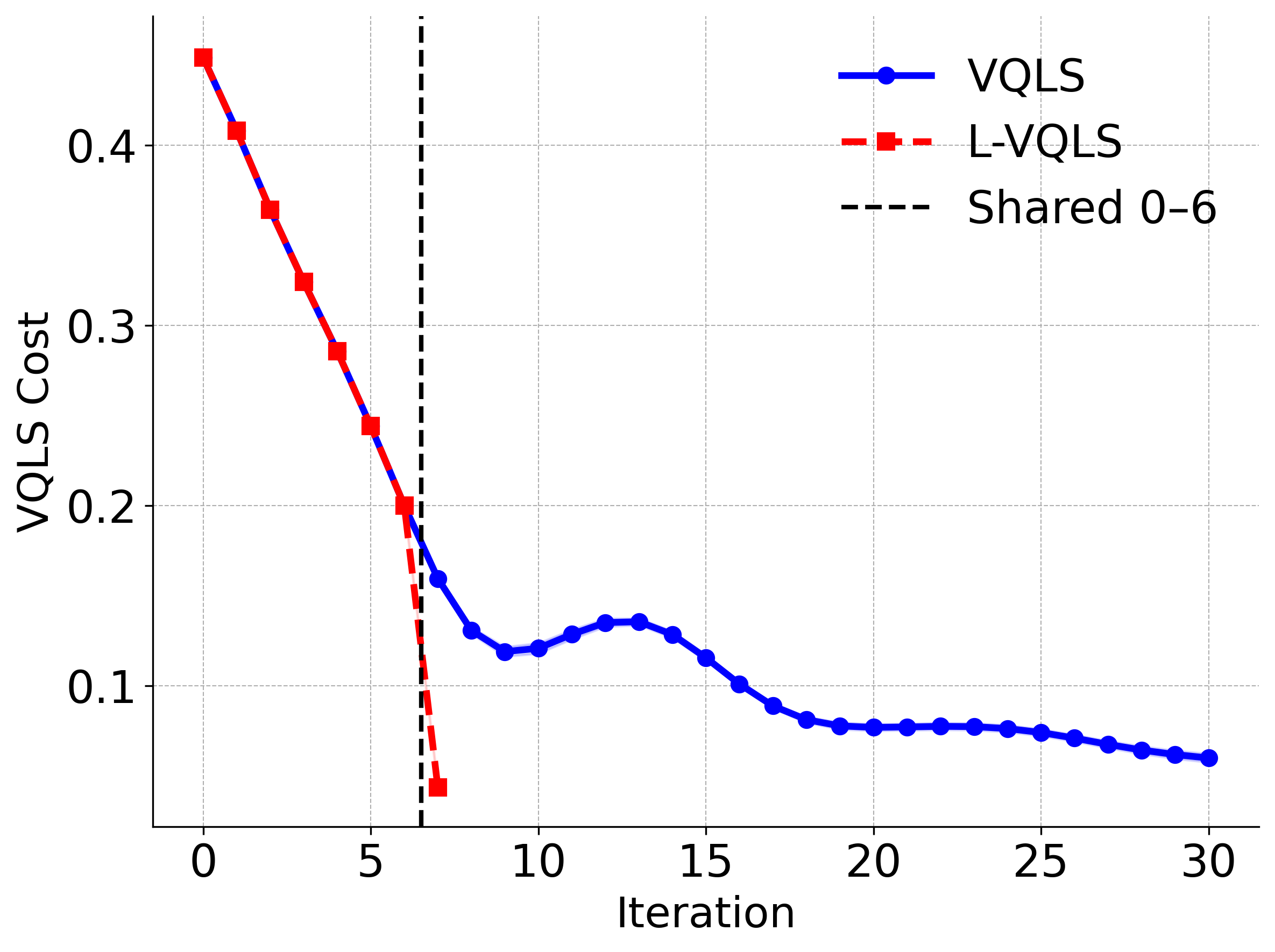}}
    \caption{Comparison of VQLS and L-VQLS cost trajectories for the 3-bus and 5-bus systems.}
    \label{fig:lvqls_cost_comparison}
\end{figure}

Fig.~\ref{fig:lvqls_cost_comparison} compares the VQLS and L-VQLS cost trajectories for representative IPM-generated systems. In both cases, the two curves share the same initial prefix. After the predicted parameter vector is applied, L-VQLS moves to a lower-cost region with fewer additional updates, while the baseline VQLS continues to decrease gradually. This indicates that the trainable model improves the remaining variational search rather than replacing the quantum linear solver.

To assess whether the reduced inner-loop effort preserves OPF-level behavior, the predicted VQLS parameters are embedded in the outer IPM loop.

\begin{figure}[!t]
    \centering
    \subfloat[3-bus system]{%
        \includegraphics[width=0.44\columnwidth]{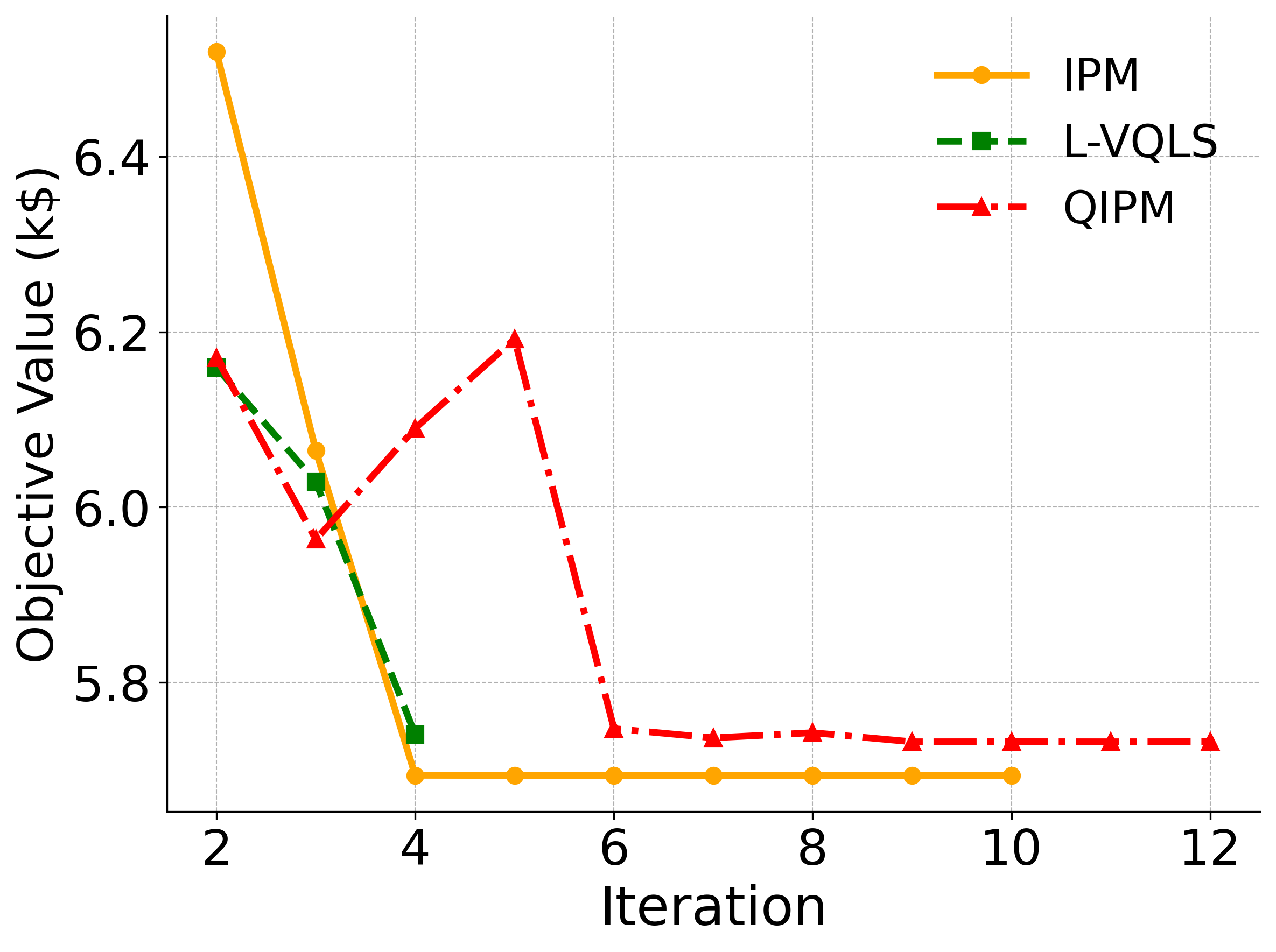}}
    \hfill
    \subfloat[5-bus system]{%
        \includegraphics[width=0.44\columnwidth]{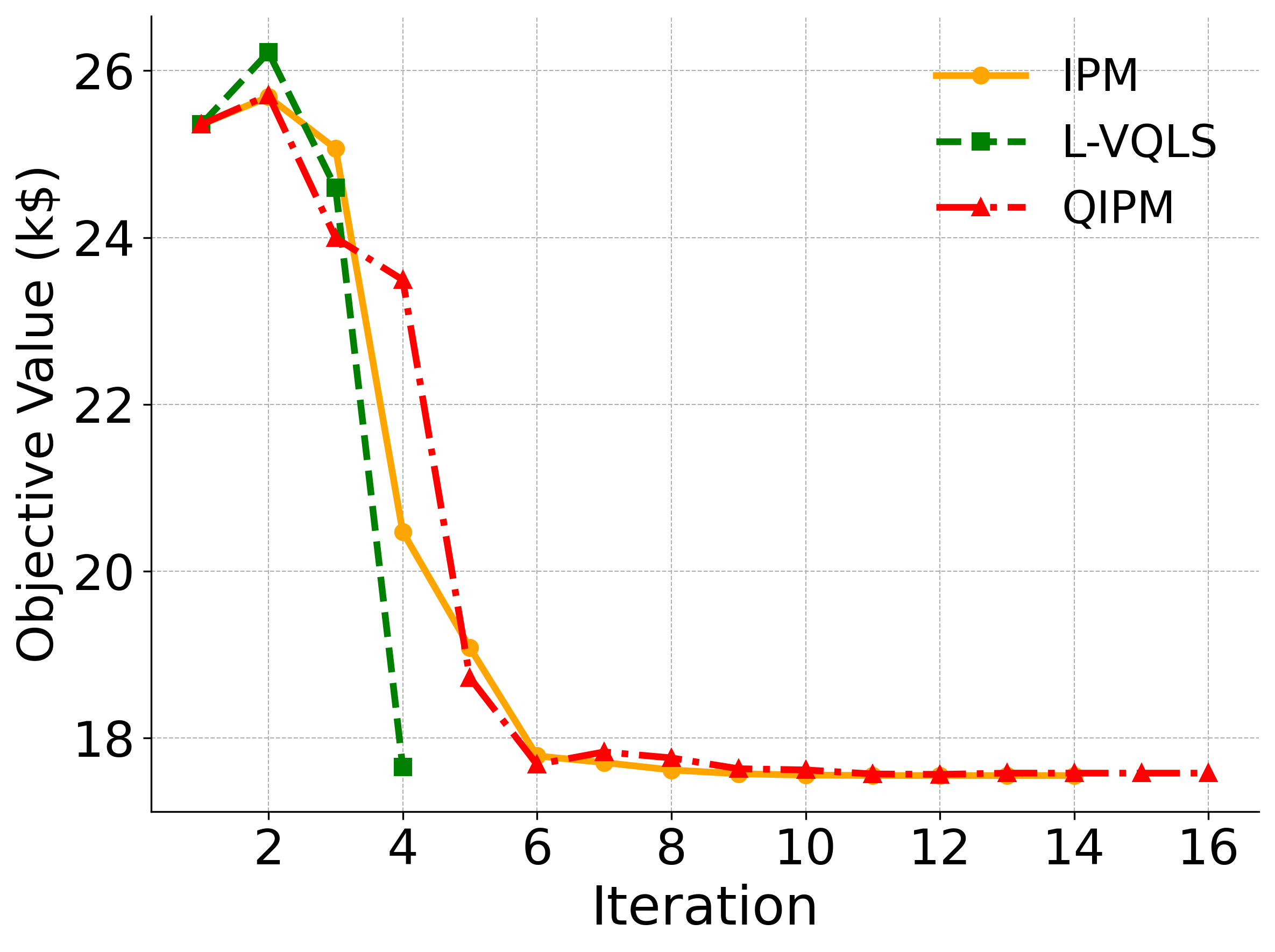}}
    \caption{OPF objective-value trajectories obtained using classical IPM, baseline QIPM, and proposed QIPM for the 3-bus and 5-bus systems.}
    \label{fig:lvqls_opf_comparison}
\end{figure}

Fig.~\ref{fig:lvqls_opf_comparison} shows that the proposed QIPM follows the general convergence behavior of the OPF solution while reaching the low-objective region in fewer outer iterations than the baseline QIPM. The classical IPM provides the reference trajectory, while the QIPM curves show the effect of VQLS-based approximate Newton directions.

\subsection{CVQLS Results}
\label{subsec:learning_cvqls_results}

The 12-bus case evaluates the proposed method on a larger IPM-generated system. The Newton-direction system has size $57\times57$, is padded to $ 64\times64$, and is converted to Hermitian form before applying the quantum routine. CVQLS is used in this case because it is more tractable for the embedded system.

The baseline CVQLS runs provide the reference variational-parameter trajectories. The proposed L-CVQLS model uses the first $K_{\mathrm{var}}^\star=7$ variational iterations as input and predicts the best-cost parameter vector observed along the baseline trajectory. This vector is then used to continue the remaining CVQLS optimization.

\begin{figure}[!t]
    \centering
    \subfloat[12-bus system: CVQLS and L-CVQLS cost trajectories]{%
        \includegraphics[width=0.44\columnwidth]{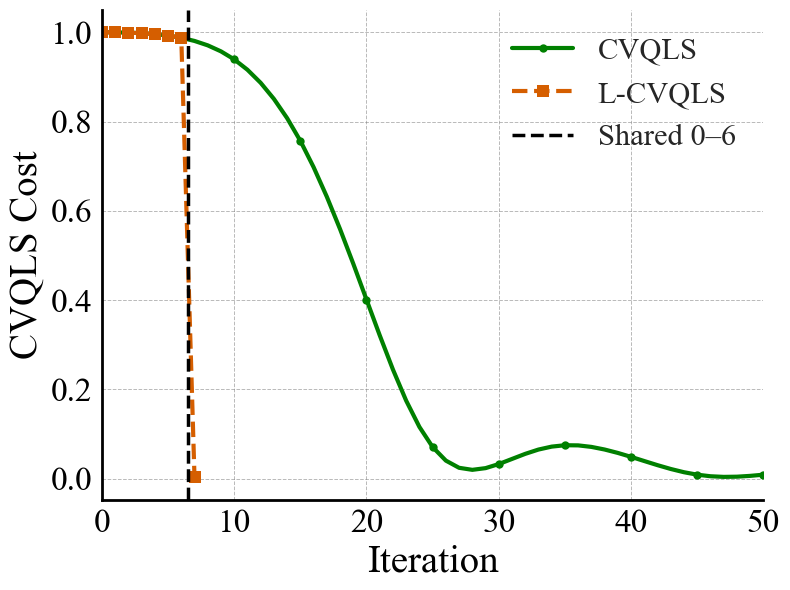}}
    \hfill
    \subfloat[12-bus system: OPF objective trajectories]{%
        \includegraphics[width=0.44\columnwidth]{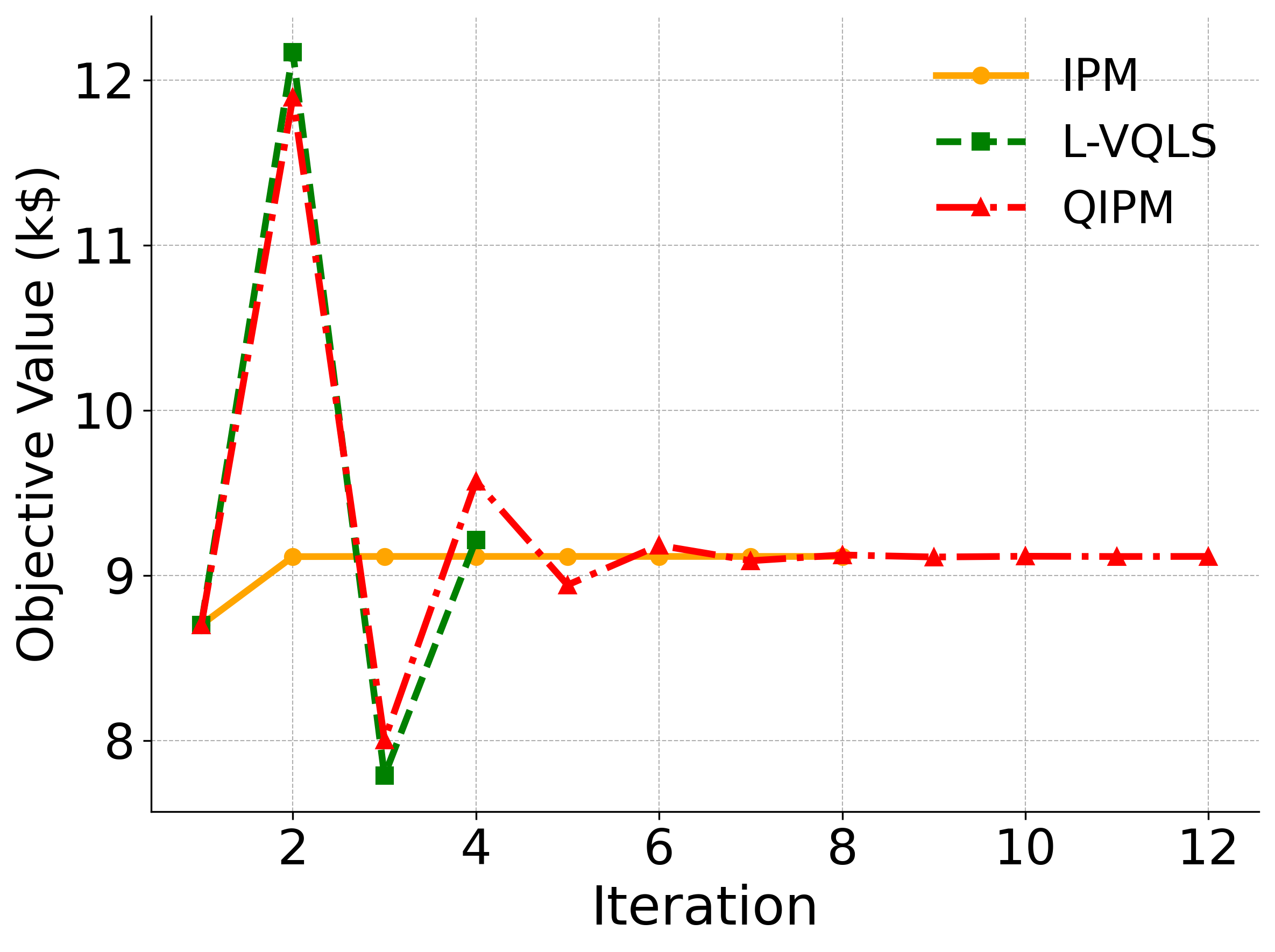}}
    \caption{Performance of the proposed L-CVQLS framework on the 12-bus system.}
    \label{fig:lcvqls_results}
\end{figure}

Fig.~\ref{fig:lcvqls_results}(a) shows that L-CVQLS moves to a low-cost region after the predicted parameter vector is applied, while the baseline CVQLS requires many additional updates. This behavior is consistent with the smaller VQLS cases.

Fig.~\ref{fig:lcvqls_results}(b) shows the corresponding OPF objective trajectories. The quantum-assisted curves are less smooth than the classical IPM reference because the CVQLS solution is more approximate for the larger embedded system. Nevertheless, both quantum-assisted methods move toward the same objective region, and the proposed version tracks the baseline QIPM behavior while reducing the inner variational effort.

\subsection{IPM-Level Central Path Projection}
\label{subsec:ipm_projection_results}

Following the validation-based selection in Section~\ref{subsec:ipm_prefix_selection}, the IPM-level model $g_{\psi}^{\mathrm{ipm}}(\cdot)$ uses the first $K_{\mathrm{ipm}}=3$ IPM iterates as the input prefix. The model predicts a later primal-dual point, which is restored using $\mathcal{R}(\cdot)$ before IPM continues. One subsequent IPM iteration is used for correction and verification.

The prediction accuracy is summarized in Table~\ref{tab:ipm_lstm_accuracy}. The high $R^2$ values and low RMSE values indicate that the model captures the central path evolution across the considered systems.

\begin{table}[!t]
\centering
\caption{Prediction accuracy of the IPM-level model.}
\label{tab:ipm_lstm_accuracy}
\scriptsize
\begin{tabular}{lcccc}
\hline
Metric & 2-bus & 3-bus & 5-bus & 12-bus \\
\hline
$R^2$ & 0.994 & 0.986 & 0.972 & 0.941 \\
RMSE  & 0.045 & 0.077 & 0.105 & 0.155 \\
\hline
\end{tabular}
\end{table}

To evaluate the OPF-level effect of the projection, the final objective value obtained by LIPM is compared with that of the classical IPM. The relative objective error is
{\small
\begin{equation}
\epsilon_f^{\mathrm{LIPM}}
=
\frac{
\left| f_{\mathrm{LIPM}} - f_{\mathrm{IPM}} \right|
}{
\left| f_{\mathrm{IPM}} \right|
}
\times 100\%,
\label{eq:lipm_objective_error}
\end{equation}}
where $f_{\mathrm{IPM}}$ is the final objective from the full classical IPM trajectory, and $f_{\mathrm{LIPM}}$ is the final objective after applying the central path projection and continuing IPM correction.

\begin{table}[!t]
\centering
\caption{Effect of IPM-level central path projection.}
\label{tab:ipm_projection_effect}
\scriptsize
\begin{tabular}{lcccc}
\hline
System & IPM iter. & LIPM iter. & Iter. red. (\%) & $\epsilon_f^{\mathrm{LIPM}}$ (\%) \\
\hline
2-bus  & 10 & 4 & 60.0 & 0.63 \\
3-bus  & 15 & 4 & 73.3 & 0.89 \\
5-bus  & 16 & 4 & 75.0 & 0.66 \\
12-bus & 10 & 4 & 60.0 & 0.85 \\
\hline
\end{tabular}
\end{table}

\begin{table}[!t]
\centering
\caption{IPM projection restoration.}
\label{tab:ipm_restoration_behavior}
\scriptsize
\setlength{\tabcolsep}{4pt}
\begin{tabular}{cccc}
\hline
Unseen cases & Nontriv. rest. & Rate & Corr. iter. \\
\hline
1000 & 11 & 1.0\% & 1 \\
\hline
\end{tabular}
\end{table}

Table~\ref{tab:ipm_projection_effect} shows that the central path projection reduces the number of outer IPM iterations to four across all systems. The resulting objective errors remain below $1\%$ with respect to the classical IPM reference. Table~\ref{tab:ipm_restoration_behavior} shows that only a small fraction of unseen cases require nontrivial restoration, and one subsequent IPM correction iteration is sufficient in those cases. These results show that the projection reduces outer-loop effort while keeping IPM responsible for correction and feasibility enforcement.

\subsection{Prediction Accuracy and Effort Reduction}
\label{subsec:prediction_effort_accuracy}

The final set of results quantifies the parameter-prediction accuracy, OPF objective error, and reduction in variational updates. The prediction error is reported for the 5-bus VQLS case and the 12-bus CVQLS case, which represent the simulated VQLS and the larger CVQLS settings.

\begin{figure}[!t]
    \centering
    \subfloat[5-bus mean]{%
        \includegraphics[width=0.47\columnwidth]{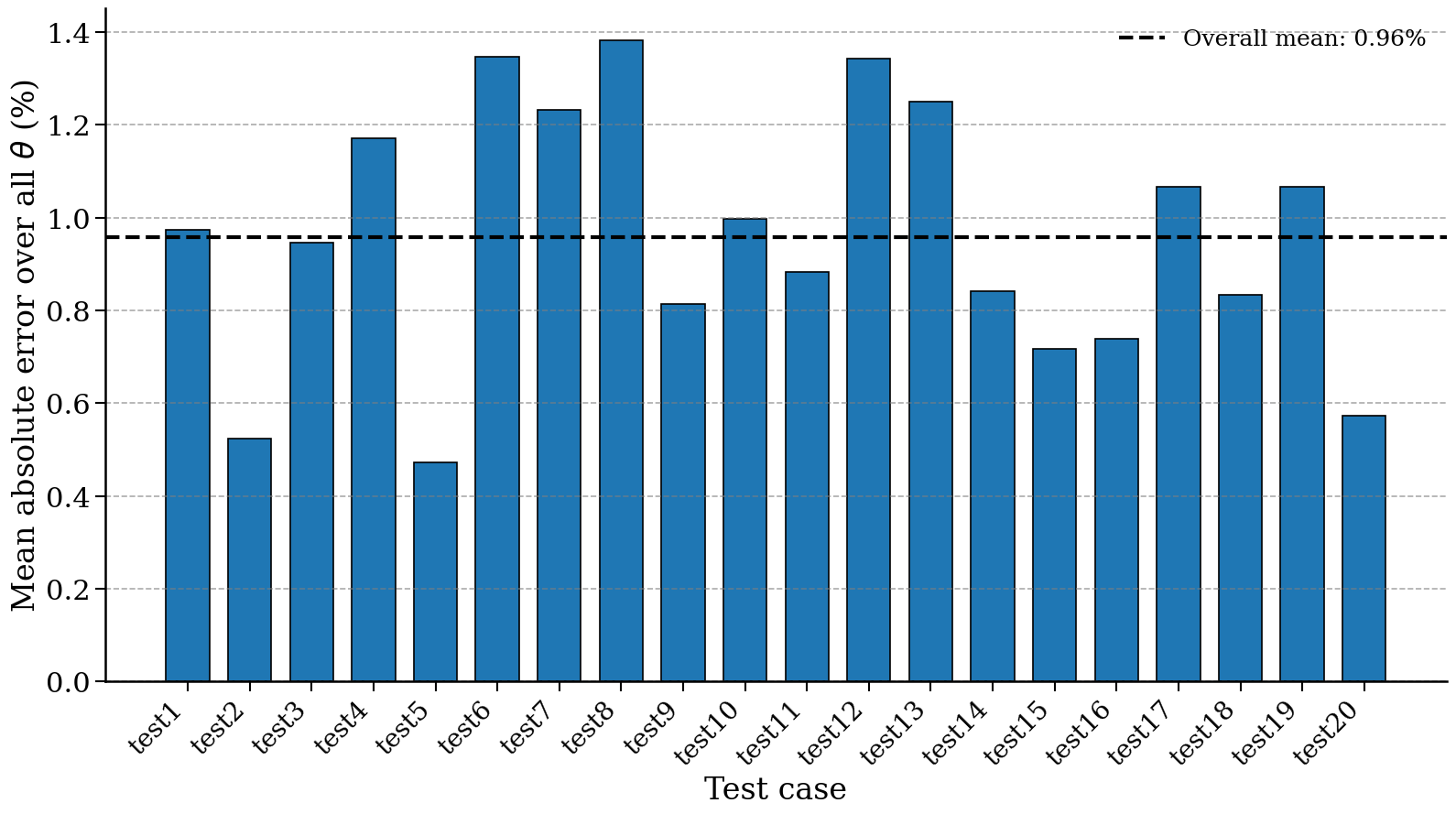}}
    \hfill
    \subfloat[5-bus selected]{%
        \includegraphics[width=0.47\columnwidth]{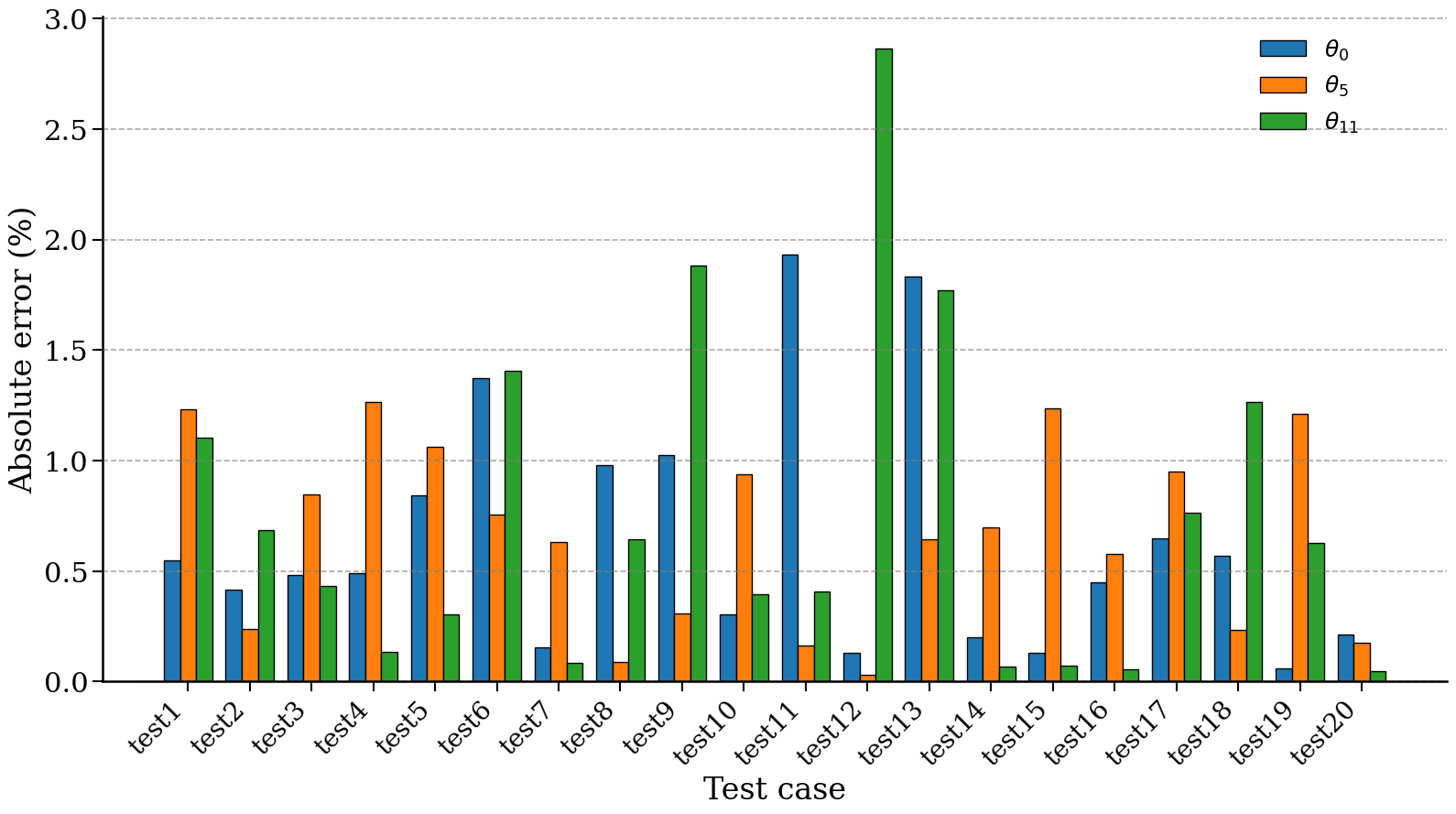}}
    \\[-1mm]
    \subfloat[12-bus mean]{%
        \includegraphics[width=0.47\columnwidth]{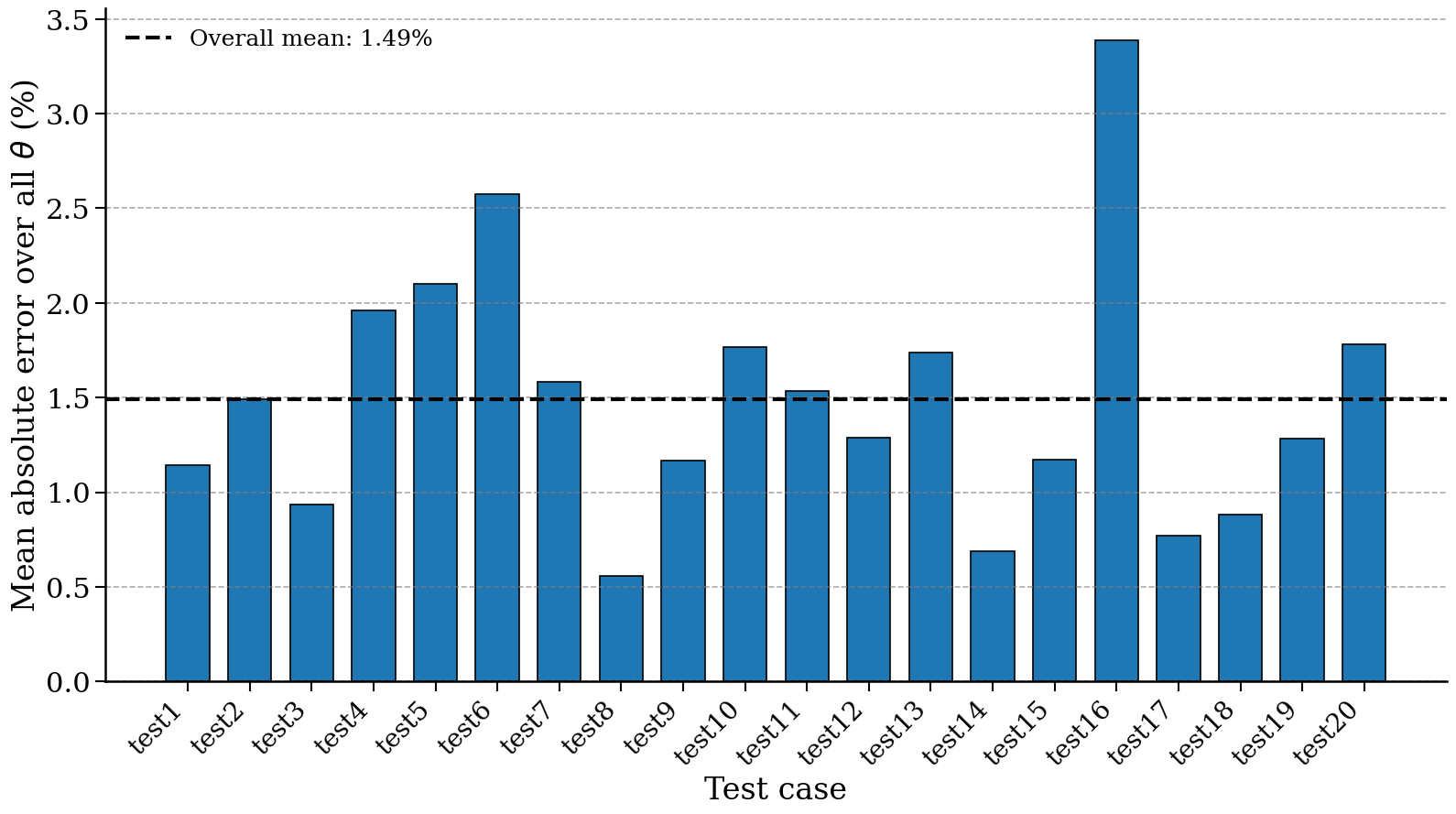}}
    \hfill
    \subfloat[12-bus selected]{%
        \includegraphics[width=0.47\columnwidth]{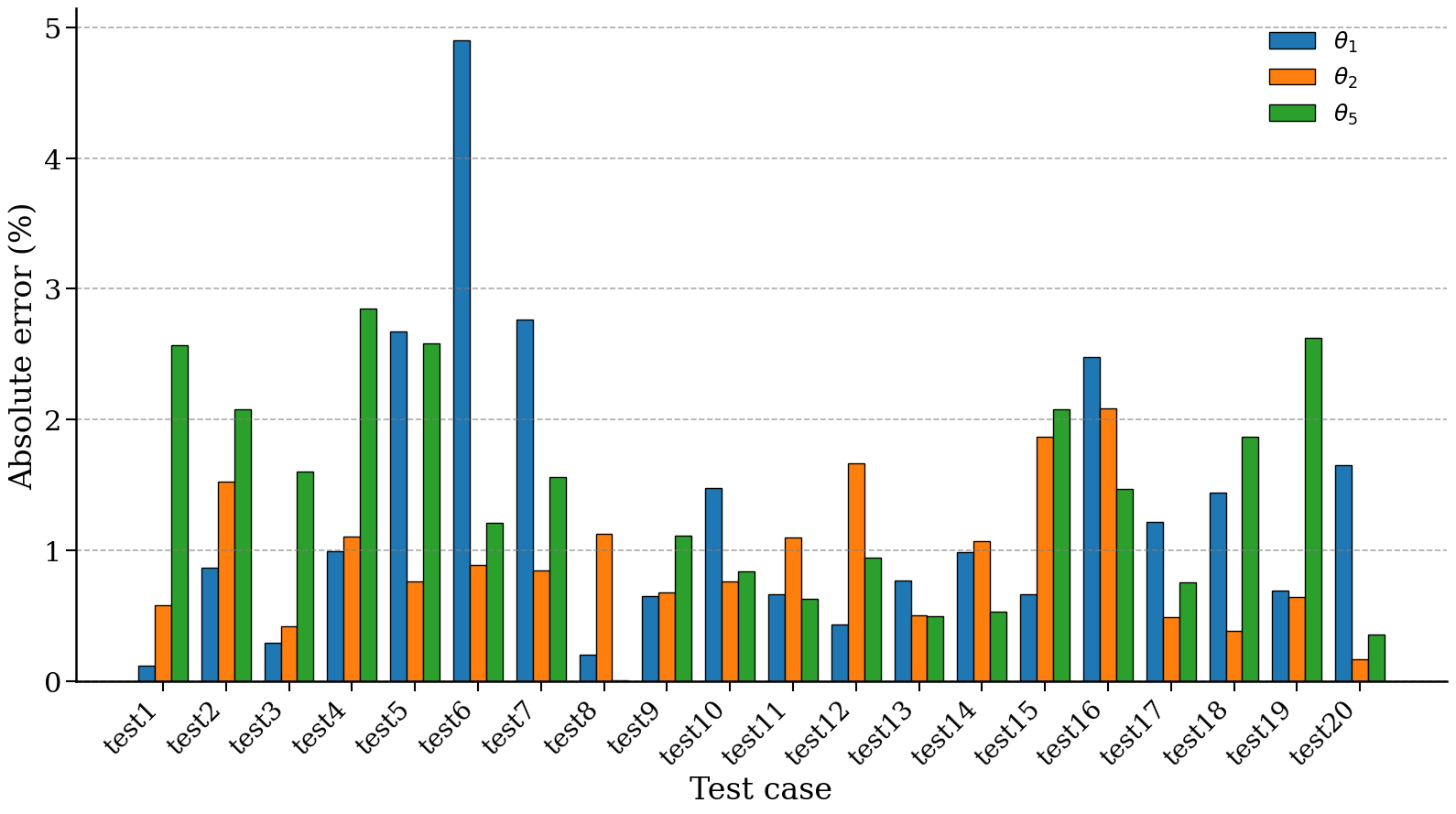}}
    \caption{Prediction error of the estimated variational parameters for the 5-bus VQLS and 12-bus CVQLS cases, averaged over 100 unseen test instances.}
    \label{fig:theta_prediction_error}
\end{figure}

Fig.~\ref{fig:theta_prediction_error} shows that the predicted parameters remain close to the baseline best-cost targets. The mean absolute prediction error is approximately $0.96\%$ for the 5-bus VQLS case and $1.49\%$ for the 12-bus CVQLS case. The larger error in the 12-bus case is expected because the embedded Newton system is larger and the CVQLS optimization landscape is more difficult.

The final OPF objective error is computed relative to the classical IPM reference:
{\small
\begin{equation}
\epsilon_f =
\frac{
\left| f_{\mathrm{method}} - f_{\mathrm{IPM}} \right|
}{
\left| f_{\mathrm{IPM}} \right|
}
\times 100\%,
\label{eq:opf_objective_error}
\end{equation}}
where $f_{\mathrm{IPM}}$ is the final objective value obtained by classical IPM, and $f_{\mathrm{method}}$ is the final objective value obtained by either baseline QIPM or the proposed QIPM.

\begin{table}[!t]
\centering
\caption{Relative OPF objective error w.r.t. classical IPM.}
\label{tab:opf_objective_error}
\begin{tabular}{lcc}
\hline
Test system & Proposed QIPM (\%) & Baseline QIPM (\%) \\
\hline
2-bus  & 2.32 & 2.00 \\
3-bus  & 0.82 & 0.68 \\
5-bus  & 0.58 & 0.16 \\
12-bus & 1.11 & 0.00 \\
\hline
\end{tabular}
\end{table}

Table~\ref{tab:opf_objective_error} shows that the proposed QIPM remains close to the classical IPM reference. The baseline QIPM yields a smaller final objective error in some cases because it continues the full variational optimization, whereas the proposed method reduces the repeated quantum-solver effort.

The reduction in variational effort is summarized in Table~\ref{tab:var_effort_reduction}. Let $N_{\mathrm{base}}^{\max}$ denote the maximum number of baseline \text{(C)}VQLS updates over the test set, and let $N_{\mathrm{learn}}$ denote the number of updates used by the proposed solver. The reduction is
{\small
\begin{equation}
R_{\mathrm{var}} =
\frac{
N_{\mathrm{base}}^{\max} - N_{\mathrm{learn}}
}{
N_{\mathrm{base}}^{\max}
}
\times 100\%.
\label{eq:var_effort_reduction}
\end{equation}}

\begin{table}[!t]
\centering
\caption{Reduction in variational optimization effort.}
\label{tab:var_effort_reduction}
\begin{tabular}{lccc}
\hline
Test system & Max baseline updates & Proposed updates & Reduction (\%) \\
\hline
2-bus  & 30  & 8 & 73.3 \\
3-bus  & 73  & 8 & 89.0 \\
5-bus  & 96  & 8 & 91.7 \\
12-bus & 165 & 8 & 95.2 \\
\hline
\end{tabular}
\end{table}

Table~\ref{tab:var_effort_reduction} shows that the proposed parameter initialization reduces the number of variational updates required by the quantum linear solver. The reduction increases with system size, reaching $95.2\%$ for the 12-bus CVQLS case. Since \text{(C)}VQLS is invoked repeatedly inside IPM, this reduction directly lowers the inner-loop computational burden of the quantum-assisted OPF framework.
\section{Conclusion}
\label{sec:conclusion}

This paper presented a nested-loop trajectory-informed quantum interior-point framework for OPF. The main motivation is that embedding \text{(C)}VQLS within IPM incurs two coupled computational costs: the number of outer IPM iterations and the number of inner variational updates per Newton solve. The proposed framework reduces both costs by using early solver-generated trajectories rather than continuing each loop without exploiting the information already available.

At the inner level, the first few \text{(C)}VQLS parameter updates provide a compact signature of the variational search direction for the current IPM-generated linear system. A trajectory-conditioned temporal learner uses this early parameter prefix to predict an informed initialization, allowing the remaining variational optimization to continue from a lower-cost region. At the outer level, early IPM iterations contain important information about the evolution of the central path, including the movement of primal variables, dual multipliers, slack variables, and the barrier parameter. A second temporal sequence model uses this prefix to project a later central path state. The projected state is restored to an admissible primal-dual point before IPM continues, so feasibility, complementarity, and optimality remain enforced by the subsequent IPM correction steps.

The framework was evaluated on 2-bus, 3-bus, 5-bus, and 12-bus OPF systems, including an IBMQ hardware demonstration for the 2-bus case. The proposed \text{(C)}VQLS-level projection reduced the number of variational updates by up to $95.2\%$, with final OPF objective deviations below $2.4\%$ relative to the classical IPM reference. The IPM-level projection reduced the outer iteration count to four across all test systems, with objective errors below $0.9\%$. Only $1\%$ of unseen cases required nontrivial restoration, indicating that the restoration step mainly acts as a safeguard.

These results show that early variational and IPM trajectories contain sufficient temporal information to guide the remaining computation. By conditioning sequence-learning models on these early dynamics, the proposed method reduces the need for repeated quantum and classical optimization while preserving the corrective role of the physics-based IPM framework. Future work will extend the method to larger networks, study robustness under hardware noise, and develop adaptive prefix-length and stopping strategies.

\bibliographystyle{ieeetr}
\bibliography{References}

@article{morstyn2022annealing,
  title={Annealing-based quantum computing for combinatorial optimal power flow},
  author={Morstyn, Thomas},
  journal={IEEE Transactions on Smart Grid},
  volume={14},
  number={2},
  pages={1093--1102},
  year={2022},
  publisher={IEEE}
}

@article{amani2025quantum,
  title={Quantum Optimization for Optimal Power Flow: CVQLS-Augmented Interior Point Method},
  author={Amani, Farshad and Kargarian, Amin},
  journal={IEEE Transactions on Smart Grid},
  year={2025},
  publisher={IEEE}
}

@article{pan2020deepopf,
  title={Deepopf: A deep neural network approach for security-constrained dc optimal power flow},
  author={Pan, Xiang and Zhao, Tianyu and Chen, Minghua and Zhang, Shengyu},
  journal={IEEE Transactions on Power Systems},
  volume={36},
  number={3},
  pages={1725--1735},
  year={2020},
  publisher={IEEE}
}

@article{amani2025optimal,
  title={Optimal power flow solution via noise-resilient quantum interior-point methods},
  author={Amani, Farshad and Kargarian, Amin},
  journal={Electric Power Systems Research},
  volume={240},
  pages={111216},
  year={2025},
  publisher={Elsevier}
}

@article{rajabi2026distributed,
  title={A Distributed Quantum Approximate Optimization Algorithm Simulator for Engineering Design Optimization},
  author={Rajabi, Ali and Hasanzadeh, Milad and Kargarian, Amin},
  journal={arXiv preprint arXiv:2606.26297},
  year={2026}
}

@inproceedings{hasanzadeh2026distributed,
  title={A Distributed Variational Quantum Eigensolver Algorithm for Unit Commitment},
  author={Hasanzadeh, Milad and Rajabi, Ali and Kargarian, Amin},
  booktitle={2026 IEEE Texas Power and Energy Conference (TPEC)},
  pages={1--6},
  year={2026},
  organization={IEEE}
}

@article{ardali2026deep,
  title={A Deep Reinforcement Learning (DRL)-Based Transformer Method for Solving the Open Shop Scheduling Problem},
  author={Ardali, Faezeh and Nyelele, Mwembezi A and Knapp, Gerald M},
  journal={arXiv preprint arXiv:2606.13682},
  year={2026}
}

@inproceedings{amani2024quantum,
  title={Quantum-inspired optimal power flow},
  author={Amani, Farshad and Kargarian, Amin},
  booktitle={2024 IEEE Texas Power and Energy Conference (TPEC)},
  pages={1--6},
  year={2024},
  organization={IEEE}
}

@article{liu2024quantum,
  title={Quantum power flows: From theory to practice},
  author={Liu, Junyu and Zheng, Han and Hanada, Masanori and Setia, Kanav and Wu, Dan},
  journal={Quantum Machine Intelligence},
  volume={6},
  number={2},
  pages={55},
  year={2024},
  publisher={Springer}
}

@article{kaseb2024quantum,
  title={Quantum neural networks for power flow analysis},
  author={Kaseb, Zeynab and M{\"o}ller, Matthias and Balducci, Giorgio Tosti and Palensky, Peter and Vergara, Pedro P},
  journal={Electric Power Systems Research},
  volume={235},
  pages={110677},
  year={2024},
  publisher={Elsevier}
}

@article{amani2025learning,
  title={Learning Interior Point Method for AC and DC Optimal Power Flow},
  author={Amani, Farshad and Kargarian, Amin and Vaidyanathan, Ramachandran},
  journal={arXiv preprint arXiv:2508.19146},
  year={2025}
}

@article{morstyn2024opportunities,
  title={Opportunities for quantum computing within net-zero power system optimization},
  author={Morstyn, Thomas and Wang, Xiangyue},
  journal={Joule},
  volume={8},
  number={6},
  pages={1619--1640},
  year={2024},
  publisher={Elsevier}
}

@article{pareek2025limitations,
  title={Limitations of Fault-Tolerant Quantum Linear System Solvers for Quantum Power Flow},
  author={Pareek, Parikshit and Jayakumar, Abhijith and Coffrin, Carleton and Misra, Sidhant},
  journal={IEEE Transactions on Power Systems},
  year={2025},
  publisher={IEEE}
}

@article{puig2025variational,
  title={Variational quantum simulation: a case study for understanding warm starts},
  author={Puig, Ricard and Drudis, Marc and Thanasilp, Supanut and Holmes, Zo{\"e}},
  journal={PRX Quantum},
  volume={6},
  number={1},
  pages={010317},
  year={2025},
  publisher={APS}
}

@article{morales2024quantum,
  title={Quantum linear system solvers: A survey of algorithms and applications},
  author={Morales, Mauro ES and Pira, Lirand{\"e} and Schleich, Philipp and Koor, Kelvin and Costa, Pedro and An, Dong and Aspuru-Guzik, Al{\'a}n and Lin, Lin and Rebentrost, Patrick and Berry, Dominic W},
  journal={arXiv preprint arXiv:2411.02522},
  year={2024}
}

@article{zimmerman2016matpower,
  title={Matpower 6.0 user’s manual},
  author={Zimmerman, Ray D and Murillo-S{\'a}nchez, Carlos E},
  journal={Power Systems Engineering Research Center},
  volume={9},
  pages={65--66},
  year={2016}
}

@article{bienstock2020mathematical,
  title={Mathematical programming formulations for the alternating current optimal power flow problem},
  author={Bienstock, Dan and Escobar, Mauro and Gentile, Claudio and Liberti, Leo},
  journal={4OR},
  volume={18},
  number={3},
  pages={249--292},
  year={2020},
  publisher={Springer}
}

@article{mastroianni2023assessing,
  title={Assessing quantum computing performance for energy optimization in a prosumer community},
  author={Mastroianni, Carlo and Plastina, Francesco and Scarcello, Luigi and Settino, Jacopo and Vinci, Andrea},
  journal={IEEE Transactions on Smart Grid},
  volume={15},
  number={1},
  pages={444--456},
  year={2023},
  publisher={IEEE}
}

@misc{CoherentVar,
  title = {{Coherent Variational Quantum Linear Solver}},
  howpublished = {\url{https://pennylane.ai/qml/demos/tutorial_coherent_vqls}}
}

@article{hasan2021hybrid,
  title={Hybrid learning aided inactive constraints filtering algorithm to enhance AC OPF solution time},
  author={Hasan, Fouad and Kargarian, Amin and Mohammadi, Javad},
  journal={IEEE Transactions on Industry Applications},
  volume={57},
  number={2},
  pages={1325--1334},
  year={2021},
  publisher={IEEE}
}

@article{schafer2007recurrent,
  author  = {Sch{\"a}fer, Anton Maximilian and Zimmermann, Hans Georg},
  title   = {Recurrent neural networks are universal approximators},
  journal = {International Journal of Neural Systems},
  volume  = {17},
  number  = {4},
  pages   = {253--263},
  year    = {2007}
}

@article{dembo1982inexact,
  author  = {Dembo, Ron S. and Eisenstat, Stanley C. and Steihaug, Trond},
  title   = {Inexact {N}ewton methods},
  journal = {SIAM Journal on Numerical Analysis},
  volume  = {19},
  number  = {2},
  pages   = {400--408},
  year    = {1982}
}

@book{wright1997primal,
  author    = {Wright, Stephen J.},
  title     = {Primal-Dual Interior-Point Methods},
  publisher = {SIAM},
  address   = {Philadelphia, PA},
  year      = {1997}
}

@article{amani2026learning,
  title={Learning optimal crew dispatch for grid restoration following an earthquake},
  author={Amani, Farshad and Ardali, Faezeh and Kargarian, Amin},
  journal={IEEE Transactions on Smart Grid},
  year={2026},
  publisher={IEEE}
}

@article{ullah2022quantum,
  title={Quantum computing for smart grid applications},
  author={Ullah, Md Habib and Eskandarpour, Rozhin and Zheng, Honghao and Khodaei, Amin},
  journal={IET Generation, Transmission \& Distribution},
  volume={16},
  number={21},
  pages={4239--4257},
  year={2022},
  publisher={Wiley Online Library}
}

@article{chen2025review,
  title={A review of quantum computing technologies in power system optimization},
  author={Chen, Yousu and Vu, Thanh Long},
  year={2025},
  publisher={Pacific Northwest National Laboratory (PNNL), Richland, WA (United States)}
}
\end{document}